\documentclass[11pt,english]{article} 

\usepackage{xcolor}
\usepackage{bbding}
\usepackage{setspace}
\usepackage{amssymb}
\usepackage[pdftex]{graphicx}
\usepackage{float}
\usepackage{amsfonts}

\usepackage{amsmath}
\usepackage{mathrsfs}
\usepackage{datetime}
\usepackage{enumerate}
\usepackage{amsthm}
\usepackage{setspace}
\usepackage{apacite}
\usepackage{etex,etoolbox}
\usepackage[latin9]{inputenc}
\usepackage{subfigure}
\makeatletter
\providecommand{\@fourthoffour}[4]{#4}
\def\fixstatement#1{%
	\AtEndEnvironment{#1}{%
		\xdef\pat@label{\expandafter\expandafter\expandafter
			\@fourthoffour\csname#1\endcsname\space\@currentlabel}}}

\makeatother

\newtheorem{theorem}{Theorem}
\newtheorem{proposition}{Proposition}
\newtheorem{lemma}{Lemma}
\newtheorem{corollary}{Corollary}

\newtheorem{conjecture}{Conjecture}

\theoremstyle{definition}
\newtheorem{definition}{Definition}

\newcommand{\argmax}{\operatornamewithlimits{argmax}}
\fixstatement{theorem}
\fixstatement{lem}
\fixstatement{proposition}
\fixstatement{corollary}

\usepackage[a4paper, margin=0.92in,includefoot]{geometry}
\usepackage{etex,etoolbox}
\usepackage{setspace}
\usepackage{natbib}
\bibpunct[, ]{(}{)}{,}{a}{}{,}%

\begin{document}

\title{Market Segmentation in Online Platforms}

\author{
	Franco Berbeglia\footnote{Tepper School of Business, Carnegie Mellon University. (fberbegl@andrew.cmu.edu)\textbf{ Corresponding author.}}
	\and
	Gerardo Berbeglia\footnote{Melbourne Business School, The University of Melbourne. (g.berbeglia@mbs.edu)}
	\and
	Pascal Van Hentenryck\footnote{Industrial and Operations Engineering \& Industrial and System Engineering, Georgia Institute of Technology. (pascal.vanhentenryck@isye.gatech.edu)}
}
\date{}
\maketitle

\begin{abstract}
	
	This paper studies ranking policies in a stylized trial-offer marketplace model, in which a single firm offers multiple products and has consumers who express heterogeneous preferences. Consumer trials are influenced by past purchases, the inherent appeal of the products, and the ranking of each product. Consumer purchases conditional on trying the product are dependent on the inherent quality for the given consumer segment. The platform owner needs to devise a ranking policy to display the products to maximize the number of purchases in the long run, and to decide whether to display the number of past purchases. The model proposed attempts to understand the impact of market segmentation in a trial-offer market with position bias and social influence. Under our model, {\color{black}consumer} choices are based on a very general choice model known as the mixed multinomial logit model{\color{black},} which embeds product appeal, ranking, and past purchases into the taste parameters. We analyze the long-term dynamics of this highly complex stochastic model and we quantify the expected benefits of market segmentation as well as the value of social influence. When past purchases are displayed, consumer heterogeneity makes buyers try the sub-optimal products, reducing the overall sales rate. We show that consumer heterogeneity makes the ranking problem NP-hard. We then analyze the benefits of market segmentation. We find tight bounds to the expected benefits of offering a distinct ranking to each consumer segment. Finally, we show that the market segmentation strategy always benefits from social influence when the average quality ranking is used. One of the managerial implications is that the firm is better off using an aggregate ranking policy when the variety of consumer preference is limited, but it should perform a market segmentation policy when consumers are highly heterogeneous.  We also show that this result is robust to relatively small consumer classification mistakes; when these are large, an aggregate ranking is preferred.

\end{abstract}

\section{Introduction}

The effects of social influence on consumer behaviour have been
observed in a wide range of settings
\citep{salganik2006experimental,tucker2011does,viglia2014please}.
Depending on the market and/or the marketing platform, social
influence may be induced by different signals, including the number of
past purchases, consumer ratings, and consumer
recommendations. Moreover, these popularity signals can be amplified
through product visibilities.  Indeed, in digital markets, the impact
of visibility on consumer behavior has been widely observed, including
in internet advertisement where sophisticated mathematical models have
been developed to determine the relative importance of the various
product positions \citep{craswell2008experimental}. Positioning
effects are also of high significance in online stores such as
Expedia, Amazon, and iTunes, as well as physical retail stores (see,
e.g., \citet{lim2004metaheuristics}).

The combination of popularity signals and product visibilities have been
extensively studied in the past few years, both theoretically and
experimentally: See, for instance,
\cite{abeliuk2016assortment,krumme2012quantifying,abeliuk2015benefits,maldonado2018popularity,van2016aligning}.
This research stream considered trial-offer markets in which consumers
have the opportunity to try products or services for free and later
decide whether to purchase them. Prominent examples of online
trial-offer markets include the following: music markets such as {\sc iTunes} where
users can listen to songs prior to the purchase decision, phone
apps stores such as {\sc Google Play}\footnote{Google Play's refund policy allows consumers to get full refunds on app purchases if the claim is made within 48 hours.} or the {\sc Apple Store} where
``apps'' have a free trial-version (limited by the functionality or
with an expiration deadline), video on demand platforms such as Netflix, where consumer try {\color{black} episodes of series} and finish watching them if they like them, and game platforms, such as Steam, where games may be returned for free within 2 hours of usage. Furthermore, other examples of trial-offer markets include online stores offering physical products that have free shipping (sometimes for purchases above a threshold) and free returns, such as Amazon.com or Nike.com. Previous studies have shown that social influence\footnote{Almost all examples of trial-offer markets show some form of popularity signal (rating and/or number of purchases) of its products when consumers browse through them.},
in conjunction with the optimization of product visibilities, can have
significant benefits on market efficiency. Moreover, simple ranking
policies, e.g., giving the most visibility to products with the best
estimated qualities, are needed to realize these benefits. These
positive results however assume that customer preferences are
homogeneous and can be modeled by a multinomial logit.

This paper studies a more realistic model with heterogeneous customers that follow a latent class multinomial logit model \citep{mcfadden2000mixed,rusmevichientong2014assortment} and attempts to understand whether the results obtained earlier in simple models would still hold under these new setting. It presents both negative and positive results. First, contrary to the homogeneous
case, this paper shows that, in mixed Multinomial Logit (MNL) trial-offer markets,
computing the ranking of products that maximizes the probability that
the next customer will make a purchase is NP-Hard under Turing
reductions. Moreover, the paper shows that popularity signals, that are
beneficial in the homogeneous case, may become detrimental to market
efficiency in mixed MNL trial-offer markets. However, this paper also
shows that these negative results can be addressed by a simple
segmentation strategy where customers are shown a quality ranking
dedicated to their own consumer segment and only observe the popularity signal
for their own market segment (i.e., the past purchases of customers of
the same segment). This market segmentation strategy can be implemented
easily by collecting information on customers and/or by providing
customers different rankings based on various demographic features.
Indeed, a recent analysis performed by the online travel agent Orbitz
has shown that Mac users spend up to about 30\% more in hotel bookings
than their PC counterparts \citep{2012_WSJ}, suggesting that it is
beneficial to show different rankings to customers depending on the
computer they use. Moreover, the major online travel agent {\sc
	Booking.com} allows users to rank hotels according to the average
consumers score of a particular segment such as couples, families, and
solo travelers. See Figure \ref{fig:awesome_image} for an illustration
of this feature. Although hotels are not a trial-offer market, the
benefits of market segmentation {\color{black} extend to trial-offer markets, as Netflix, Google Play and Steam, among others, implement targeting strategies, giving different product recommendations based on past consumer behavior \citep{2019_medium,2020_WSJ}.}


\begin{figure}[!t]
	\centering
	\includegraphics[width=0.8\textwidth]{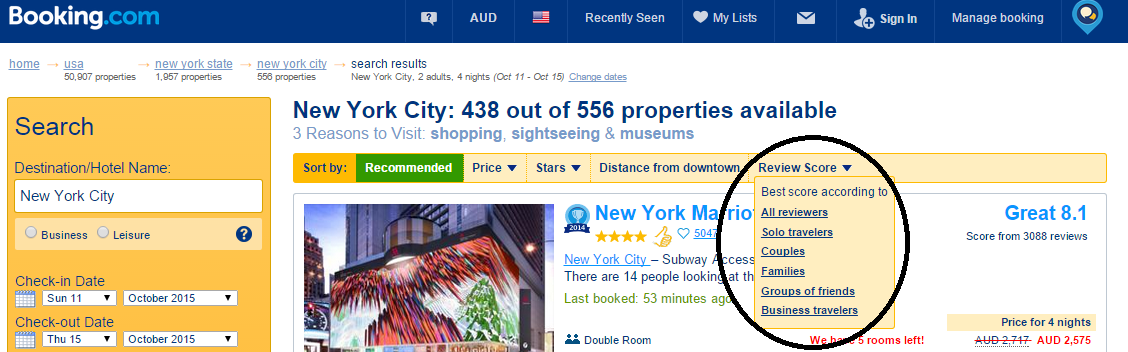}
	\caption{Rankings Averaging Scores of a Consumer Segment at Booking.com.}
	\label{fig:awesome_image}
\end{figure}

The contributions of this paper can be summarized as follows:

\begin{enumerate}
	
	\item The paper constructs the first trial-offer market model with
	social influence and position biases in which consumers preferences
	can follow any finite mixture of multinomial logits.
	
	\item The paper shows that, in mixed MNL trial-offer markets,
	computing the ranking of products that maximizes the probability
	that the next customer will make a purchase is NP-Hard under Turing
	reductions.
	
	\item The paper shows that the popularity signal may, under some
	circumstances, decrease the expected market efficiency. In other
	words, the display of past purchases may reduce the number of sales
	by confusing consumers about which products to try.
	
	\item The paper studies the average quality ranking, which ranks the
	items in decreasing order of average quality. It shows that the
	average quality ranking converges to a unique equilibrium when
	consumers are shown the number of past purchases (the popularity
	signal). This proof is rather involved as it requires to show that
	the mixed MNL model can be seen as a special MNL model in which some
	parameters (appeal and quality) are no longer constants but
	functions of the past purchases vector, and that these quantities
	can be upper and lower bounded in order to demonstrate convergence.
	
	\item The paper presents a simple segmentation strategy, where
	customers are shown a quality ranking dedicated to their own segment
	and only observe the popularity signal for their own market segment
	(i.e., the past purchases of customers of the same segment). The paper
	quantifies the potential benefits in market efficiency of this
	strategy.  Specifically, it proves that the expected purchases can
	increase up to a factor of $K$, where $K$ is the number of segments
	of the mixed MNL model.
	
	\item These theoretical results are complemented by a series of
	computational experiments which provide several managerial insights
	about trial-offer markets.
\end{enumerate}

\noindent
The remaining of this paper is organized as follows. Section
\ref{related_literature} reviews the literature most related to this
work. Section \ref{the_model_section} introduces the model of the
dynamic trial-offer market. The most relevant ranking policies for
this model are described in Section \ref{ranking_policies_section},
which also presents the NP-hardness results for performance ranking in
mixed multinomial logit models. Section \ref{section:quality}
describes the convergence and the impact of social influence for the
quality ranking in the same setting.  Section
\ref{market_segmentation_section} presents our segmentation strategy
and its benefits. Section \ref{section:simulations} presents results
of computational experiments and Section \ref{section:conclusion}
concludes the paper. The proofs are deferred to the appendix.
In the supplementary appendix A, we consider an extension of the
model in which the platform owner makes mistakes during the customer
classification process. The supplementary appendix B shows that a local search heuristic for solving
the ranking optimization problem does not bring much benefit with
respect to the much less costly average quality ranking policy studied
in Section \ref{section:quality}. Supplementary appendix C analyzes an extension in which the firm's objective is to maximize the expected revenue instead of purchases. Finally, the supplementary appendix D analyzes an extension in which the platform owner may show a subset of products to all consumer segments.

\section{Related literature} \label{related_literature}
Our work is related to the MusicLab experiment performed by
\citet{salganik2006experimental}. In that experiment, participants
were presented a list of unknown songs from unknown bands, each song
being described by its name and band. The participants were
partitioned into two groups exposed to two different experimental
conditions: the {\em independent} condition and the {\em social
	influence} condition.  In the independent group, participants were
shown the songs in a random order and they were allowed to listen to
each of them and then download them if they wish.  In the second group
(social influence condition), participants were shown the songs in
popularity order, i.e., allocating the most popular songs to the most
visible positions. Moreover, these participants were also shown a
popularity signal, i.e., the number of times each song was downloaded
too. In order to investigate the impact of social influence,
participants in the second group were distributed in eight ``worlds''
evolving completely independently. In particular, participants in one
world had no visibility about the downloads and the rankings in the
other worlds. The MusicLab is an ideal experimental example of a
trial-offer market where each song represents a product, and listening
and downloading a song represent trying and purchasing a product
respectively. The results by \citet{salganik2006experimental} show
that the different worlds evolve significantly differently from one
another, providing evidence that social influence may introduce
unpredictability in a market.

To explain these results, \citet{krumme2012quantifying} proposed a
framework in which consumer choices are captured by a multinomial
logit model whose product utilities depend on songs appeal, position
bias, and social influence. \citet{abeliuk2015benefits} provided a
theoretical and experimental analysis of such trial-offer markets
using different ranking policies following the framework of
\citet{krumme2012quantifying}. They proved that social influence is
beneficial in order to maximize the expected number of purchases when
using a greedy heuristic known as \emph{performance ranking}. The
\emph{performance ranking} selects the ranking that maximizes the
expected number of purchases at the next time period, i.e. it
maximizes the short-term market efficiency. \citet{abeliuk2015benefits}
have also illustrated experimentally that the popularity ranking \footnote{The popularity ranking ranks (dynamically) products by the number of purchases in decreasing order.} is
outperformed by the performance ranking in a variety of
settings. Still based on the model of \citet{krumme2012quantifying},
\citet{van2016aligning} have studied the performance of the
\emph{quality ranking} which ranks products by their intrinsic quality
(the quality of a product is here defined as the probability that a
consumer would purchase/download the product once she has tried the
product out). They show that the quality ranking is in fact
asymptotically optimal and has a considerably less unpredictability
than the popularity ranking.

\cite{maldonado2018popularity} studied the impact of the popularity signal
strength on a market with multiple products and social influence. In
their model, the popularity signal strength is a an exogenous parameter
$r>0$. The authors provide a complete characterization of the
long-term market share of each of the products and show that the
market is completely predictable as long as $r \leq 1$.

The relative importance of different popularity signals have been recently
investigated by \citet{engstrom2018demand} and \citet{viglia2014please}. The first
paper focuses on how consumers choose apps in the {\sc Google Play}
platform, and the second one studies how people select hotels. Both
experiments arrived to the same conclusion, namely that the popularity
signal (i.e., the number of purchases) has a much stronger impact on
consumer behavior than the average consumer rating signal.

Our work is related to the recent paper by \citet{hu2015liking} who
consider a monopolist facing a newsvendor problem with two
substitutable products with the same quality in which consumer
preferences are affected by past purchases. The authors showed that
the market is unpredictable but it can become less so if one of
products has an initial sales advantage (such as for example by
providing an initial discount). Our model has considerable differences
including the incorporation of position biases, highly richer consumer
preferences (Mixed MNL), an arbitrary number of products, and the
allowance of products to have difference qualities.

\citet{ghose2012designing} proposed a ranking system for hotels which
takes into account the economic value of different locations and
service-based characteristics, as well as consumer heterogeneity. In a
follow-up paper, \citet{ghose2014examining} studied the effects of
three ranking policies on consumer behaviour using archival data
analysis and randomized experiments. The general idea of both papers
is to build a simultaneous equations model of clickthrough, conversion
(purchase decision), ranking (performed by the platform owner), and
customer rating.  In their model, the demand for the different hotels
(i.e., the products) is independent between the different hotels
(apart from the fact that each hotel is assigned a different position
in the ranking) whereas, the model we study in this paper, the
different products are in direct competition to attract the demand.

In another related paper, \citet{gopal2016design} perform a
quantitative study on how a firms can strategically alter malleable
networks such as enterprise social networks (ESN) or consumer social
networks (CSN) in order to increase the transmission of ideas,
innovation, or other information. In our setting, we focus on comparing
a complete network (consumers observe all purchases) versus a network
where consumers are partitioned in $K$ classes or clusters (and a
consumer only observes purchases of individuals from their own
cluster).

\citet{vaccari2018social} studied a model consumers where arrive in sequence
and estimate the quality of products based on product reviews (likes
and dislikes). In their model, consumers like a product if the
product's quality exceeds their expectation (which is calculated based
on past ratings). The authors provide conditions that allow consumers
to learn the true quality of products in the long run. Our work is also related to recent studies of theoretical choice models that incorporate position biases. As in our model, they consider situations in which the probability of selecting a product does not only depend on the offer set but also on the way products are displayed \citep{abeliuk2016assortment, aouad2015display, davis2013assortment, gallego2020approximation}.

Recently, \citet{golrezaei2018two} considered a similar framework
to ours in which the firm needs to decide how to rank a set of
products to sell to consumers. Similar to our setting, sorting
products by utility in their model is not always optimal \footnote{Although the underlying reasons for it under each of the models are different.}. Moreover, finding the optimal ranking of products to maximize short-term sales (or revenue)
can be found in polynomial time when there is a single consumer type,
but the problem becomes NP-hard when the number of consumer types is
more than one (We note that a first preprint of our paper containing all
these results was posted on November, 2015 (ArXiv)
\citep{berbeglia2015benefits}). In their work, the authors
constructed the choice probabilities from a two-stage consumer search
model based on a seminal work by \citet{weitzman1979optimal} on the Pandora's problem. As a result, the resulting
mathematical expression for the choice probabilities is different from
the one we study. Another difference is that \citet{golrezaei2018two}
considers the problem of maximizing welfare (not studied here) whereas
a fundamental focus of this work is on segmentation strategies which
are presented in Section \ref{market_segmentation_section}.

Another recent paper that models the dynamics of customers influenced by social influence is due by \cite{chen2019reviews}. An important difference to our work is that social influence signal are product ratings rather than purchases, and that their choice model is based on the MNL model rather than on any finite mixture of MNLs.

Other papers related to our work are \citet{lee2020discovering}, who uses sales transaction data to estimate the parameters of a Mixed MNL through which are able to identify heterogeneous consumers groups, \citet{zhen2019effects}, \citet{capuano2017fuzzy}, and \citet{molinero2015cooperation} who studied social influence models, and \cite{lutz2018consumer} and \cite{anderson2014pricing} who focused on market segmentation of customers.

\section{The Model} \label{the_model_section}

\paragraph{Motivation} We consider a firm running a marketplace that
sells a set of products. Following
\citet{krumme2012quantifying} and \citet{salganik2006experimental}, we focus our
attention to trial-offer markets, i.e., markets in which consumers can
try the product for free before deciding to make a purchase. Consumers
are position-biased in the following sense: {\color{black}the} likelihood of trying a
specific product is affected by the position of the product, as well
as the position of the other products in the market. We also
consider that in this marketplace it is possible to display information
about product popularity. In particular, we assume that the firm shows
the total number of purchases for each product at each point in
time. 

Unlike \citet{krumme2012quantifying}, we consider that there are
different segments of consumers. More precisely, the probability that a given
consumer tries a product will follows a Mixed Multinomial Logit
(MMNL). Since \citet{mcfadden2000mixed} proved that every random
utility model can be well approximated by a MMNL, our model of
consumer preferences is indeed very general.

\paragraph{Formalization} We now formally describe a dynamic model for
this marketplace. Let $[N]=\{1,2,\hdots,N\}$ denote the set of items
in the marketplace and $S_N$ denote the set of the permutations of
these items. At any point in time, the firm decides how to position
the items in the market by selecting a permutation $\sigma \in S_N$
such that $\sigma(i)=j$ implies that item $i$ is placed in position
$j$ ($j \in [N]$).

The consumer behavior can be described as follows. There are $K$
different segments of consumers. At any point in time the probability that the next arriving customer belongs to segment $k\in[K]$ is given by $w_k$, the segment's {\em weight}\footnote{A special case of this is a Poisson arrival process with arrival rate $\lambda_k$ for each consumer segment $k\in[K]$, such that $w_k=\frac{\lambda_k}{\sum_j \lambda_j}$.}. Consumers from the same segment exhibit different purchase profiles due to idiosyncratic shocks. This is captured with the MNL model.

When consumer $t$ enters the market, she observes all the items and a
popularity vector $d^t=(d_1^t,d_2^t,\hdots,d_N^t) \in \mathbb{N}^N$,
where $d_i^t$ is the number of times item $i$ has been bought for
prior to her arrival at time $t$. When the consumer arrives, each of
the $N$ items have been given a position through a permutation $\sigma
\in S_{N}$. The consumer selects an item to try and then decides
whether to buy it. Following \citet{krumme2012quantifying} and extending their model to multiple segments, if the consumer belongs to segment $k$, the
probability that she tries item $i$ is given by
\begin{equation}
\label{MMNL}
p_{i,k}(\sigma,d^t) =  \frac{v_{\sigma(i)} (a_{i,k} + d_i^t)}{\sum_{j=1}^N v_{\sigma(j)} (a_{j,k} + d_j^t)+z_k}
\end{equation}
where $z_k \in \mathbb{R}_{\geq0}$ for all $k\in [K]$  is fixed to a constant for the
duration of the process, $v_j \in \mathbb{R}_{\geq 0}$ represents the
visibility of position $j \in [N]$ regardless of the consumer class
(the higher the value $v_j$ the more visible the item in that position
is), and $a_{i,k} \in \mathbb{R}_{>0}$ captures the intrinsic appeal
of item $i$ for consumer segment $k$ for all $i \in [N]$ (higher values
correspond to more appealing items)\footnote{Note that the expression in Equation \eqref{MMNL} is the special case of a MMNL when the consumer utilities are logarithmic.}. The value $z_k$ represents the outside option for those consumers of segment $k$ that enter the market. As the total number of purchases increases, the fraction of consumers choosing the outside option will decrease.

If a consumer from segment $k$ has selected item $i$ for a trial, the
probability that she would purchase the item is given by $q_{i,k} \in
[0,1]$. Observe that this probability is independent of both the
appeal vector $(a_{1,k},\hdots,a_{N,k})$ and the visibility vector
$v$. Intuitively, this assumption which has been validated in the
MusicLab experiment, captures the fact that it is more difficult to
influence consumers after they have tested a product than before.

When the consumer decides to purchase item $i$, the popularity/sales
vector $d$ is increased by one in position $i$. To analyze this
process, we divide time into discrete periods such that each new
period begins when a new consumer arrives. Hence, the length of each
time period is not constant.

The objective of the firm running this market is to maximize the total
expected number of purchases. To achieve this, the key managerial
decision of the firm is what is known as the ranking policy
\citep{abeliuk2015benefits}, which consists in deciding at each point in
time the permutation $\sigma \in S_N$ to display the items.  The next
section describes a number of relevant ranking policies for this
model.

A key aspect of this paper is to study the potential benefits of the
popularity signal in terms of the rate of purchases (market
efficiency) and compare the ranking policies with and without this
signal. In this paper, we always assume that the popularity signal is
used as specified in Equation \eqref{MMNL}. When the popularity signal
is not used, the probability of trying (or sampling) a product is
obtained as if the popularity signal is simply the vector $\langle
0,\ldots,0 \rangle$.

We conclude this section with a several comments about the model described
above. First, under the trial-offer market model, a customer who tries, but
does not like, a product simply walks away without trying any other
product. An alternative model would be to allow the customer to try
another product with probability $c_i$ every time she tried but did
not like a product $i$. Perhaps surprisingly, one can show that under
mild conditions\footnote{The probability $c_i$ is a polynomial
	function on the product quality $q_i$.}, the resulting model is
equivalent to the original trial-offer market with a single trial per
consumer \citet{van2016trial} \footnote{The equivalence is based on a
	redefinition of product appeal and quality.}. Second, although there are many different ways to model consumer trial and purchase probabilities as a function of the previous purchases, intrinsic appeal, etc; the way we model the problem in this paper is a natural extension to an MNL model that has been empirically tested \citep{krumme2012quantifying}: it has succeeded to fit the data from a large scale randomized experiment where over $14,000$ users listened and downloaded songs from an assortment 48 song pieces. The participation in the experiment was unpaid and voluntary and therefore, the setup is arguably closer to a genuine internet platform than if individuals were paid to participate (see \citet{salganik2006experimental} for more details). {\color{black}Third, the model studied here could potentially be useful in markets that do not have the trial-offer structure. Indeed, this model can be captured as a function whose inputs are (1) market observations (current product ranking and past product purchases) and (2) some parameters subject to estimation (product appeals, product qualities, and position visibility values). The function then returns each product purchasing probability. Thus, even for non-trial-offer markets, it is an open empirical question on whether this model is better or worse than other choice models that also consider product rankings and/or popularity signals such as \citet{golrezaei2018two} and \citet{vaccari2018social}, in terms of predictive accuracy of the purchasing probabilities.} Finally, it is worth observing that our model assumes there is an independence between trying a product and buying it. Although this assumption was not problematic to fit the large-scale experiment carried out in \citet{salganik2006experimental}, a model extension that incorporates a correlation between the two steps would be interesting.


\section{Ranking Policies}
\label{ranking_policies_section}

Consider without loss of generality that the $N$ locations are sorted
by their visibility such that $v_1 \geq v_2 \geq \hdots \geq v_N$. A
ranking policy is a function $f: \mathbb{N}^N \to S_N$ which,
given a vector of past purchases, returns a ranking of the items.

Ranking policies can be partitioned into two groups: {\em static} and
{\em dynamic}. A ranking policy $g$ is said to be \emph{static} if the
output ranking does not depend on the popularity signal, i.e., if
$f(d) = f(d')$ for all $d, d' \in \mathbb{N}^N$. On the other hand, a
\emph{dynamic} ranking policy is one in which the output ranking
depends on this signal.

\subsection{Performance ranking}
%

The \emph{performance ranking} is a dynamic policy that greedily
selects a ranking that maximizes the expected number of purchases in
the following period. This strategy was first proposed by
\citet{abeliuk2015benefits} for the special case with $K=1$ (where they show it is asymptotically optimal) and we now
generalize its definition for the more general model considered in
this paper. The probability that the next incoming consumer belongs to segment $k$
is given by $w_k$, therefore the performance ranking at time period $t$ consists of
finding the permutation $\sigma^* \in S_N$  maximizing the
probability of a purchase in the next time period, i.e.,
\begin{align}
\label{performance_ranking-MMNL}
\sigma^* = \argmax_{\sigma \in S_N} \sum_{k=1}^K w_k \cdot \sum_{i=1}^N p_{i,k}(\sigma,d^t) \cdot q_{i,k}.
\end{align}
The probability $\Pi^{PR}$ of a purchase in the next time period is
thus given by
\begin{align}
\label{performance_ranking}
\Pi^{PR} & = \max_{\sigma \in S_N} \bigg\{\sum_{k=1}^K \Big( w_k \cdot \sum_{i=1}^N (p_{i,k}(\sigma,d^t) \cdot q_{i,k}) \Big) \bigg\}\\
& = \max_{\sigma \in S_N} \bigg\{ \sum_{k=1}^K \Big( w_k  \cdot \sum_{i=1}^N \Big(\frac{v_{\sigma(i)} (a_{i,k} + d_i^t)}{\sum_{j=1}^N v_{\sigma(j)} (a_{j,k} + d_j^t)+z_k}  \cdot q_{i,k} \Big) \Big) \bigg\}.
\end{align}

\noindent
\citet{abeliuk2015benefits} showed that when $K=1$, this greedy ranking
policy can be computed efficiently, i.e., in strongly polynomial time
(Theorem 1 in
\citet{abeliuk2015benefits}, they assumed
$z=0$ but their proof can be easily generalized for any $z \in
\mathbb{R}_{\geq 0}$). Moreover, despite the myopic focus of the
performance ranking, a series of computational experiments performed
by \citet{abeliuk2015benefits} showed that, for the special case of
$K=1$, the performance ranking was superior than the standard
popularity ranking both in terms of unpredictability as well as in
terms of number of purchases.  Unfortunately, the performance ranking
cannot be computed efficiently when there are at least two classes of
consumers. More precisely, we can show that the assortment problem
under a 2-segment Mixed Multinomial Logit choice model which is known to
be NP-hard \citep{rusmevichientong2014assortment} can be reduced
(under Turing reductions) to computing the performance ranking in our
setting.

\begin{theorem}\label{theorem_hardness}
	Computing the performance ranking is NP-hard under Turing
	reductions. This is true even when $K=2$ and the product qualities
	are the same for all consumer classes.
\end{theorem}

In order to deal with this negative result, this paper explores two
options. The first is to rank products based on their average quality,
a strategy we called {\em average quality ranking} and is studied in
the next section.  The second avenue is to explore a greedy local
search heuristic which worked as follows. In the first period, we set
the average quality ranking as initial solution, and then, we evaluate possible ways to exchange the position of two products in the
ranking (2-swaps) and select the first swap that improves the objective function.  The local search process
is then repeated until there are no more improvements.  In each
consecutive period, we use the previous period ranking as the starting
solution. Notice that this heuristic will finish in a finite number of steps. This is because the objective function increases with each swap and there exist finitely many assortments of products. Our experimental results have shown that the additional
benefit of the local search is minimal (see Appendix C), we thus used
the average quality ranking as our approximate method which is several
orders of magnitude faster.

\subsection{Average quality ranking} \label{section:quality}
The \emph{average quality ranking} (or quality ranking for short) is a
simple and natural static policy: it consists in ranking the products
by the weighted average quality (among the different customer
classes), ignoring the appeals and popularity signal. The quality
ranking for the MMNL model thus consists in placing in position $j$
the item with the $j^{th}$ highest weighted average quality, where the
weighted average quality of item $i \in [N]$ is
\begin{align}
\bar{q}_{i} &= \sum_{k=1}^K w_k q_{i,k}.
\end{align}
For the special case $K=1$, the quality ranking is optimal asymptotically and always benefits
from the popularity signal used in our model \citep{van2016aligning}. The next section will
study whether this continues to hold in richer contexts when $K>1$. Note that, in the
following, the ranking which orders the products by decreasing values
of $\bar{q}_{i}$ is called the \emph{average quality ranking}.


Before analyzing some fundamental properties of the MMNL it is important to first make the following definition.
\begin{definition}
	The MMNL model \emph{goes to a monopoly} using a ranking policy $f$
	if, for each realization of the $N$ random sequences
	$\{\phi_i^t\}_{t \in \mathbb{N}}$ ($i \in [N]$), there exists a
	product $i^*$ such that the realized sequence $\{\phi_{{i^*}}^t\}_{t
		\in \mathbb{N}}$ converges almost surely to $1$ as $t$ goes to infinity. In this
	case, we also say that item $i^*$ \emph{goes (predictably) to a
		monopoly}.
\end{definition}

We now study some fundamental properties of the quality ranking for
the MMNL model. We first show that the average quality ranking
converges to a monopoly (under weak conditions). We then study the
benefits of displaying the popularity signal for the average quality ranking.


Given a ranking policy $f$, the random variable
\[
\phi_i^t=\frac{d_i^t}{\sum_j d_j^t}
\]
is known as the market share of item $i$ at time $t$: It represents
the ratio between the number of times that item $i$ was purchased and
the total number of purchases up to time $t$.

\noindent
We can now show that the MMNL model goes to a monopoly when using the
average quality ranking. The proof is quite technical: Its key idea is
to show that the Mixed Multinomial Logit Model (MMNL) can be reduced
to a generalized case of the Multinomial Logit Model (MNL) where the
appeal and the quality of an item at time $t$ depend on the popularity
signal at $t$ (This is shown in Theorem \ref{monopoly_alltogether}). Then this generalized appeals and qualities can be bounded to obtain the result. The proof relies on the following lemma that
generalizes the convergence result of the quality ranking for the MNL
model by \citet{van2016aligning} to the case where the
appeal and quality of an item depend on the popularity ranking
provided that the resulting functions are bounded by above and
below. In Theorem \ref{monopoly_alltogether} we show that there exists a time period $t^*$ after which the time dependent qualities and appeals of this modified MNL model are bounded as in \ref{modifiedMNL_bounds} and that these bounds satisfy \ref{max_reduced_quality}. Thus, by using Lemma \ref{qsuc}, we can show that this modified MNL model goes to a monopoly which implies that the MMNL goes to a monopoly.

\begin{lemma}
	\label{qsuc}
	Consider a different Multinomial Logit Model, i.e., a setting with $K=1$ where
	the appeal and quality of each item $i$ are functions of the purchases
	vector $d^t$, i.e., $\widetilde{a}_i^t=\widetilde{a}_i(d^t)$ and
	$\widetilde{q}_i^t=\widetilde{q}_i(d^t)$ respectively. Suppose that
	there exists a time period $t^*$ such that these two quantities are
	upper and lower bounded by constants for any period $t>t^*$
	independent of the realizations of $\widetilde{a}_i^t$ and
	$\widetilde{q}_i^t$, i.e.,
	\begin{equation}\label{modifiedMNL_bounds}
	q_{i,min}\leq \widetilde{q}_i^t\leq q_{i,max} \text{ and } a_{i,min}\leq \widetilde{a}_i^t\leq a_{i,max} \quad\forall i \in [N], t>t^*.
	\end{equation}
	Let $\sigma \in S_N$ denote a static ranking policy. If there exists an item $i^*$ and an instant $\hat{t}$ such that
	\begin{equation}\label{max_reduced_quality}
	v_{\sigma(i^*)}q_{i^*,min}>v_{\sigma(i)}q_{i,max} \text{ } \forall \text{ } i\neq i^* \text{ and } \forall \text{ } t>\hat{t},
	\end{equation}
	then item $i^*$ goes to a monopoly when using the ranking policy
	$\sigma$.
\end{lemma}


\noindent
The main result of this section is about the convergence to a monopoly
of a large class of static ranking policies. For simplicity, we assume
a weak condition to break potential ties between items.

%

\begin{definition}
	A static ranking policy $\sigma$ is \emph{tie-breaking} for a MMNL
	model if there exists a unique item $i^*$ with the highest product
	of visibility and weighted average quality, i.e.,
	\begin{align}
	& \exists \  i^*\in[N]: \bar{q}_{i^*}v_{\sigma(i^*)}>\bar{q}_iv_{\sigma(i)} \quad \forall  i \in [N], i \neq i^* \text{} \label{cond1}.
	\end{align}
\end{definition}

\noindent
Note that this tie-breaking property is a very mild assumption: If the
quality for each pair (consumer class, product) is a realization from
a (different or the same) continuous probability distribution over
some interval (regardless of how small), the probability that a
ranking policy is tie-breaking is indeed 1. We are now ready to prove
the main result of this section.

\begin{theorem}
	\label{monopoly_alltogether}
	Consider a MMNL model $\mathcal{M}$ and a static, tie-breaking ranking
	policy $\sigma \in S_N$ for $\mathcal{M}$. Model $\mathcal{M}$ goes to
	a monopoly using $\sigma$ and the item $i^*$ that goes predictably to
	a monopoly using $\sigma$ in $\mathcal{M}$ is given by
	\[
	i^*=\argmax_{1\leq i \leq N}{v_{\sigma(i)}\bar{q}_i}.
	\]
\end{theorem}

\noindent
The following corollary asserts that the average quality ranking
converges to a monopoly for the product of highest average
quality. This result is particularly interesting since it shows that
the average quality ranking is asymptotically optimal, and that it generalizes the quality ranking from the
MNL to the MMNL model. By asymptotically optimal we mean that it is the global ranking with global social influence that maximizes the purchase probability as time goes to infinity, $\lim_{t\rightarrow \infty}  P^t $, where $P^t$ is the expected purchase probability at period $t$.

\begin{corollary}(Asymptotic Optimality of the Average Quality Ranking).
	\label{cor1}
	Whenever the average quality ranking is used, a MMNL model goes to a monopoly
	for the product with the highest weighted average quality.
\end{corollary}

\noindent
We end this section with some comments about the implications of
Theorem \ref{monopoly_alltogether} and Corollary \ref{cor1}. {\color{black} In general, market monopolies are not a desirable outcome from a consumer perspective. A key aspect of Theorem \ref{monopoly_alltogether} is that the monopoly outcome does not come as a result of a restriction on the product offer variety (every products is always offered in our model), but because the consumer's utility ratio between the most popular product and the other product utilities tends to infinity.} Although such monopoly convergence is somewhat surprising, it naturally applies to the very long run dynamics of these trial-offer markets
(i.e., when time tends to infinity). In practice, however, it takes a
very long time to even get close to a monopoly. It is not simple to find the convergence rate to a monopoly, even for the case $K=1$. This model can be seen as an unbalanced and irreducible Pólya Urn Process for which the convergence rate is an open problem except for a few special cases of replacement matrices (see remark 4.7 in \cite{janson2004functional}, where they conjecture a convergence rate of $o(n^\gamma)$ with $\gamma=\min\lbrace{1/2,3(1/2-Re(\lambda_2/\lambda_1)\rbrace})$, with $\lambda_1$ and $\lambda_2$ being the greatest and second greatest eigenvalues of the replacement matrix respectively). As an illustration of the
dynamics, we have performed a variety of computational experiments which are
reported in Section \ref{section:simulations}. In all those
experiments, the convergence to a monopoly of a single product was not
yet achieved at the end of the simulation. From a practical point of
view, the main insight obtained from Theorem
\ref{monopoly_alltogether} and Corollary \ref{cor1} is that, as time
goes by, consumers are more likely to purchase the product that has
the greatest average quality, which is due to the effect of the popularity
signal in the market dynamics. An interesting question that arises is
whether the effect of the popularity signal is beneficial to the
firm. Specifically, what is the effect of the popularity signal on the
expected rate of products purchased? This question is addressed in the
following section.

\subsubsection{The Impact of the Popularity Signal}
In the previous subsection, we have shown that the average quality
ranking for the MMNL model inherits the asymptotic convergence of the
quality ranking for the MNL. Under the MNL model, the probability (in
expectation) that the next individual purchases some product is always
increasing if the popularity signal is used
\citep{van2016aligning}. In other words, the information of past
purchases is helping consumers to make better choices on which
products to try, meaning that they are more likely to purchase
them. Unfortunately, this result does not always hold when consumers
follow the more general MMNL model.

\begin{theorem}
	When using the average quality ranking, the MMNL model can perform
	(1) up to $K$ times better if the popularity signal is not shown,
	where $K$ is the number of classes; and (2) can perform arbitrarily
	worse without showing the popularity signal than by showing it.
	\label{bound-nsi}
\end{theorem}

\noindent
Theorem \ref{bound-nsi} provides an upper bound on how much the
expected sales per period can be reduced due to the impact of the
popularity signal (part (1)). In addition, it shows that sometimes, showing the social influence can be extremely beneficial (part (2)). For the latter, it is important to note that not showing the popularity signal can perform arbitrarily worse than showing it even when $z=0$. This implies that this reduction is not caused by having consumers switching from the outside option to other products, but because consumers under social influence can sometimes try popular products that have low quality for them (relatively to other products). In Proposition \ref{bound-nsi-tightness} we show that the bounds in Theorem $3$ are tight, and this happens when we choose $z=0$ (however, for different values of $z$, not showing the popularity signal can still perform arbitrarily worse than by showing it).

\begin{proposition}
	The bounds in Theorem \ref{bound-nsi} are tight.
	\label{bound-nsi-tightness}
\end{proposition}
%

\noindent
This result shows that, depending on the specific parameters of the
trial offer market (i.e. consumer's appeals, product qualities and
weights of the different consumer segments), the popularity signal can
enhance the number of sales or it can be detrimental to
them. Specifically, the expected number of purchases under the average
quality ranking policy can decrease by a factor of $K$ if the
popularity signal is used (for some settings) but they can also be
increased by an arbitrarily large factor in other settings. In
Proposition \ref{bound-nsi-tightness}, the first tight bound occurs in
a trial-offer market where each of the $K$ different classes of
consumers have a unique product in which they are interested;
Moreover, this product is different for each consumer segment and
there is a perfect alignment between product appeal and product
quality for each class. In such a setting, the
global popularity signal would be detrimental for the firm as well as for the consumers. The reason is that in the long run product 1 will become a monopoly (since its average quality ranking is higher than all the others). This means that in the long run, consumers will tend to try product 1. But this is problematic because, except for one consumer segment, all consumers segments do not like this product. In the limit, only a $1/K$ fraction of the consumers would purchase this product. {\color{black} In short, in this scenario a global popularity signal will persuade consumers to try products they do not like.}

On the other hand, the second bound in Proposition
\ref{bound-nsi-tightness} was obtained in a similar market setting but
in which there is negative correlation between the products appeal and
their quality. In that setting, the market converges slowly to a
monopoly in the case the popularity signal is used. If no social
influence is used, the rate of purchases tends to zero as the products
have essentially no appeal. Thus, an heterogeneous set of customers
complicates the managerial decisions in the marketplace: Whether or
not social influence is beneficial in trial-offer markets depends on
the particular structure of preferences among consumers. We recall
that this is in sharp contrast to the results in
\citep{van2016aligning} where it is shown that the popularity signal
is always beneficial in the case consumers share preferences (i.e
$K=1$).
We now quantify the benefit (or cost) of showing the popularity signal given the specific model parameters. The asymptotic purchase probability ratio between the average quality ranking with no social influence (AQNSI) and the average quality ranking with social influence (AQGSI) is (see the proof of Theorem \ref{bound-nsi}):

\begin{equation}\label{eq:GSI-NSI}
\frac{\lim_{t\to \infty}P^t_{AQNSI}}{\lim_{t\to \infty}P^t_{AQGSI}}=\frac{\sum_{k=1}^K w_k \sum_{i=1}^N q_{i,k}\frac{v_{\sigma(i)}a_{i,k}}{\sum_{j=1}^N v_{\sigma(j)}a_{j,k}+z_k}}{\max_{1\leq i \leq N}{\sum_{k=1}^K w_k q_{i,k}}}.
\end{equation}

If the expression in \ref{eq:GSI-NSI} is smaller (greater) than $1$  showing the popularity signal under the average quality ranking is beneficial (detrimental).  Proposition \ref{bound-nsi-tightness} shows that under different model parameters, the ratio in Equation \eqref{eq:GSI-NSI} could be equal to zero, $K$, or any number between them. This is important for the platform owner, and has to be decided particularly for each market settings, since social influence may hurt or benefit purchases in the long run.

\section{Market Segmentation and its Benefits}
\label{market_segmentation_section}
In the previous section, we have shown a number of negative results
for the MMNL model. In particular, we have shown that, in MMNL models,
computing the performance ranking is intractable and that
displaying the popularity signal to customers may significantly reduce
the asymptotic market efficiency (i.e. the expected rate of
purchases) of the average quality ranking. In this section, we show
that the widely used marketing strategy known as \emph{market
	segmentation} remedies these limitations, while retaining the
original benefits of quality ranking for the Multinomial Logit Model.

The market segmentation considered here assumes that the firm has the
ability to know the segment of each arriving consumer. This is a natural
assumption in a number of online markets (e.g., Amazon, online retail
stores, iTunes, Google Play and Netflix) where firms are able to learn information
about their customers over time. Armed with this information, the firm
will now propose item rankings dedicated to each customer
segment. Moreover, and equally important, the popularity signal will be
tailored to each segment. In other words, {\em the firm will only show
	the popularity signal derived from purchases of customers of the
	same segment as the incoming customer}, not the popularity obtained
from the entire customer pool. As shown in Figure
\ref{fig:awesome_image}, websites such as {\tt Booking.com} already
give customers the option of selecting their peer groups to refine the
site recommendations (although hotels bookings are not a trial-offer
market, segmentation is an important factor in their revenue
maximizing strategies). Under this new strategy where each consumer segment has its own quality ranking and observes the past purchases
of its own segment only, the policy is called the {\em segmented quality
	ranking}. The firm uses $K$ permutations $\sigma_k \in S_N (k \in
[K])$, where $\sigma_k$ sorts the products in decreasing order
according to their quality for consumer segment $k$. In addition, the probability
of trying item $i$ for a customer of segment $k$ is given by
\[
p_{i,k}(\sigma,d_k^t)
\]
where $d_k^t = (d_{1,k}^t,\ldots,d_{N,k}^t)$ and $d_{i,k}^t$ denotes
the number of purchases of item $i$ by customers from segment $k$ up to
time $t$.

We now study the benefits of this market segmentation. Observe first
that each market segment can be viewed as evolving independently and
hence directly inherits the original benefits identified for the
quality ranking under the MNL model: The market share of the highest
quality product converges asymptotically to 1.  This observation will
enable us to quantify the benefits of market segmentation. We begin by
providing some key definitions that will be required later.

\begin{definition} The segmented quality ranking policy $\sigma_k \in S_N$
	$(k \in [K])$ is \emph{tie-breaking} if, for each segment $k$, there
	exists a unique item $i^*_k$ with the highest quality:
	\begin{align}
	\forall k \in [K] \; \exists i^*_k\in[N] \; \forall j \in [N], j \neq i^*_k: \; q_{i^*_k,k} > q_{j,k}. \label{cond_locally_without_ties}
	\end{align}
\end{definition}

\noindent
Assuming a MMNL model for which the average quality ranking and the
segmented quality ranking are tie-breaking, we compare the probability
of a purchase at time $t$ in both settings. More precisely, we compare
two quantities:
\begin{itemize}
	
	\item $P_{AQGSI}^t$: the probability of a purchase at time $t$ when
	the firm uses the average quality ranking and the ``global'' popularity signal $d^t$;
	
	\item $P_{SQSSI}^t$: the probability of a purchase at time $t$ when
	the firm uses the segmented quality ranking with the consumer segment's popularity signal $d_k^t$.
\end{itemize}

\noindent
The probabilities $P_{AQGSI}^t$ and $P_{SQSSI}^t$ concern the
behavior of the consumer arriving to the market at time $t$
independently from the customer class. Comparing $P_{AQGSI}^t$ and
$P_{SQSSI}^t$ for any time $t$ is a very challenging task. Instead, we
compare both variables in the limit. The following theorem shows that
the market segmentation strategy is always beneficial for the firm and
the benefit it provides is upper bounded by a factor of $K$.

\begin{theorem}
	\label{thm_bound}
	Assume that the average quality ranking and its segmented version are
	tie-breaking for a MMNL model. Then,
	\begin{equation}
	\label{benefits_clustering}
	1 \leq \lim_{t\rightarrow\infty}  \frac {P^t_{SQSSI}}{P^t_{AQGSI}} \leq K.
	\end{equation}
\end{theorem}
%
%
%
%

The following proposition shows that there are settings in which the
(multiplicative) benefits of market segmentation are indeed equal to
the upper bound provided in Theorem \ref{thm_bound}.

\begin{proposition}
	The upper bound of Theorem \ref{thm_bound} is tight.
	\label{tightness_segmentation}
\end{proposition}

\noindent
These results show that the segmented quality ranking always
outperforms (or it is equal to) the average quality ranking in
expectation, and that the improvement in market efficiency can be up to
a factor of $K$. As it can be seen from the proof of Proposition
\ref{tightness_segmentation}, the market settings in which the
segmentation strategy beats by the most the global ranking is when
each consumer segment $k \in \{1,\hdots,K\}$ of consumers have a single product
with a non-zero quality and it is pairwise different between any two
classes. Thus, the different segments of the market have very distinct
preferences, which is the setting where the strategy of market
segmentation is generally used in practice \citep{dickson1987market}.
Proposition \ref{tightness_segmentation} means that there are settings under which a segmented ranking performs the same, or up to $K$ times better than a single ranking. This imposes a maximum cost rate to what the platform owner may be willing to pay for information about consumer segments, and deploying a segmented ranking. It is important to remark that these results hold for the very long run
(i.e. asymptotic in nature).  They don't necessarily imply that the
segmented quality ranking is always better, since the popularity
signal is weaker early in the market evolution. This will be
illustrated in the computational experiments presented in the next
section.

It is important to remark that the results of this section can be
extended to the case where there are classification mistakes while
performing the segmentation policy. This situation is analyzed in the
supplementary appendix A. We formulate the problem using a mistake probability
matrix and show that the system converges to a monopoly, analogous to
Theorem \ref{monopoly_alltogether} but incorporating the values of the
error probability matrix. Assuming that classification mistakes occur
with an exogenous probability $\beta_0$ evenly distributed among all
products, we find an analogous bound as Theorem \ref{thm_bound}, which
becomes
$K(1-\beta_0)$. For
an example, if we make classification errors $10\%$ of the time with 5
consumer classes of equal weights, the maximum benefit of segmentation
is $5(1-0.1)=4.5$, while its maximum benefit is $5$ with
perfect classification.  This extension generalizes our model to a
more realistic setting where we can analyze the trade-off between
segmenting the market and taking the risk of making classification
errors versus being risk averse and showing an average quality
ranking. We performed a numerical simulation to analyze the impact of
having different mistake probabilities for the SQSSI policy and we
find that the average quality ranking AQGSI outperforms the SQSSI
policy when $\beta_0>10\%$ for a parameter set (see Figure \ref{fig:mistakes_scheme1} in Appendix A).

We end this section with two additional comparisons. First, we compare how the average quality ranking without social influence (AQNSI) performs in comparison to the segmented quality segmented social influence ranking (SQSSI). It is no surprise that SQSSI outperforms AQNSI as the number of purchases goes to infinity: this is because under SQSSI, in the limit, each consumer segment will try the product that maximizes its purchase probability. Finally, we analyze how SQSSI performs in comparison to its counterpart without social influence (SQNSI). Again here, it should be no surprise that SQSSI outpeforms SQNSI. Both of these results are shown in the following theorem.

\begin{theorem}
	\label{AQNSIvsSQSSI_bound}
	The ratio between the asymptotic purchase probability of the average quality or the segmented quality ranking without social influence (AQNSI or SQNSI), and its segmented quality ranking with segmented social influence (SQSSI) is always less than 1,
	\begin{equation*}\label{eq:SI-NSI}
	\frac{\lim_{t\to \infty}P^t_{AQNSI}}{\lim_{t\to \infty}P^t_{SQSSI}}\leq 1 \text{ and } \frac{\lim_{t\to \infty}P^t_{SQNSI}}{\lim_{t\to \infty}P^t_{SQSSI}}\leq 1.
	\end{equation*}
\end{theorem}

We have shown that SQSSI outperforms AQGSI, AQNSI and SQNSI, and that AQGSI can outperform or underperform AQNSI. Figure \ref{pictogram_policies} shows a pictogram with the ranking policies studied.

\begin{figure}[h]
	\includegraphics[trim=0cm 3cm 0cm 0cm ,width=17cm]{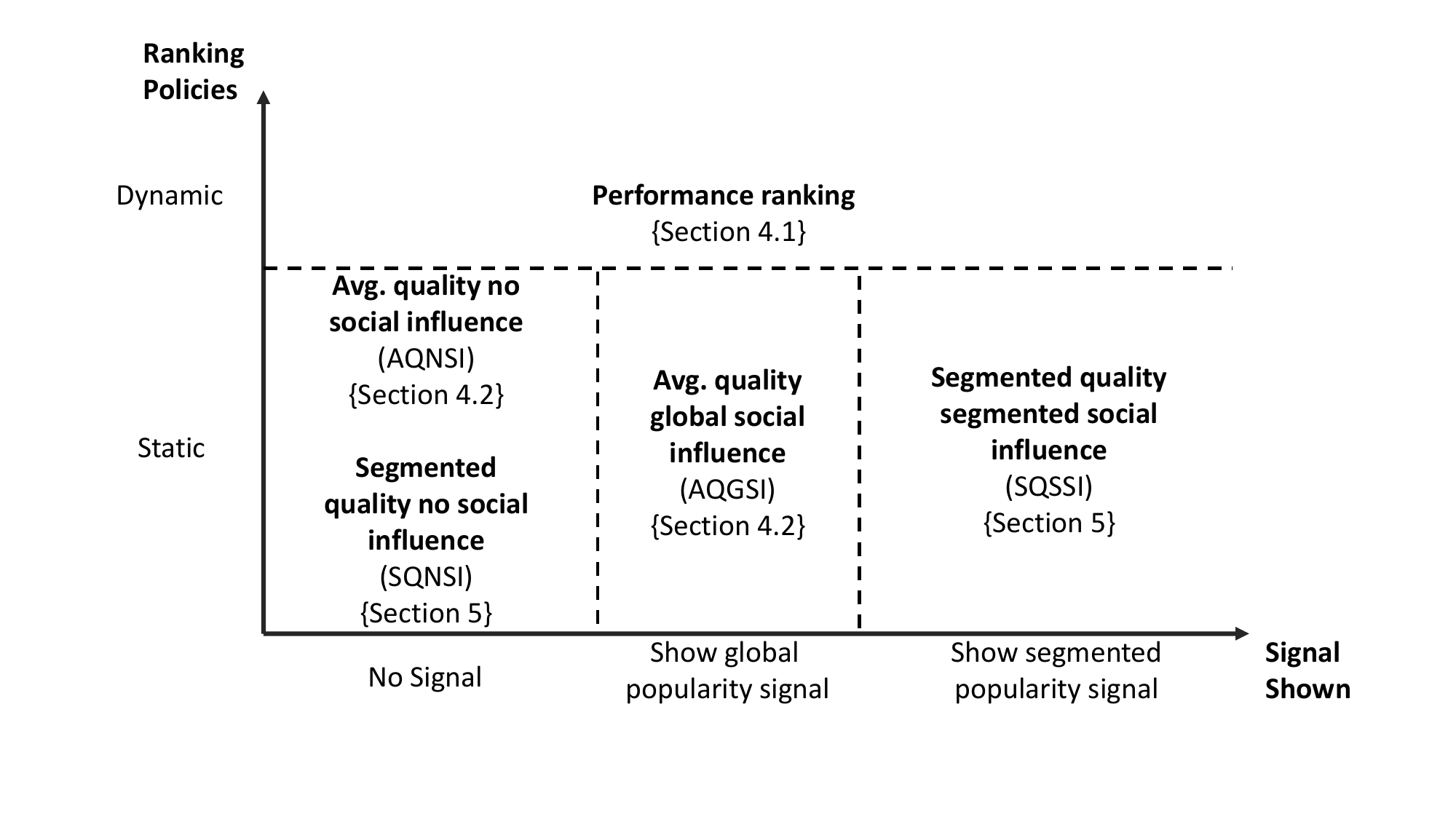}
	\caption{\label{pictogram_policies}Pictogram of the ranking policies studied and where to find them.}
\end{figure}

\section{Computational Experiments}
\label{section:simulations}

This section presents the results of computational experiments to
illustrate the theoretical results and complement them by depicting
how the markets evolve over time for different types of rankings.

\subsection{The Experimental Setting}

\paragraph{The Agent-Based Simulation} The experimental setting uses
an agent-based simulation to emulate the MusicLab
\citep{salganik2006experimental}. It generalizes prior results which
simulated the MusicLab through the use of a MNL model (e.g.,
\citep{krumme2012quantifying,abeliuk2015benefits}) to a MMNL model.  Each
simulation consists of $T$ iterations and, at each iteration $t$ $(1
\leq t \leq {\color{black}T}$), the simulator
\begin{enumerate}
	
	\item randomly selects a customer segment $k$ according to the classes
	weights $w_k$;
	
	\item randomly selects an item $i$ for the incoming customer according
	to the probabilities $p_{i,k}(\sigma,d)$, where $\sigma$ is the ranking
	proposed by the policy under evaluation and $d$ is the popularity
	signal;
	
	\item randomly determines, with probability $q_{i,k}$, whether
	the selected item $i$ is purchased. In the case of a purchase, the
	simulator increases the popularity signal for item $i$, i.e.,
	$d_{i,t+1} = d_{i,t} + 1$. Otherwise, $d_{i,t+1} = d_{i,t}$.
	
\end{enumerate}

\noindent
The experimental setting aims at being close to the MusicLab
experiments and it considers $50$ items and simulations with $T=200,000$
steps. The reported results in the graphs are the average of $400$
simulations. The analysis in \cite{krumme2012quantifying} indicated
that participants are more likely to sample products with better
ranking positions. More precisely, the visibility decreases with the
ranking position, except for a slight increase at the bottom
positions. To have a fair comparison between the settings with and without social influence we set $z=0$ for all simulations, in that way the fraction of customers choosing the outside option does not change over time.

\paragraph{Qualities and Appeals}

\begin{figure}[t]
	\centering
	\includegraphics[trim = 30mm 77mm 30mm 77mm, clip, width=12cm]{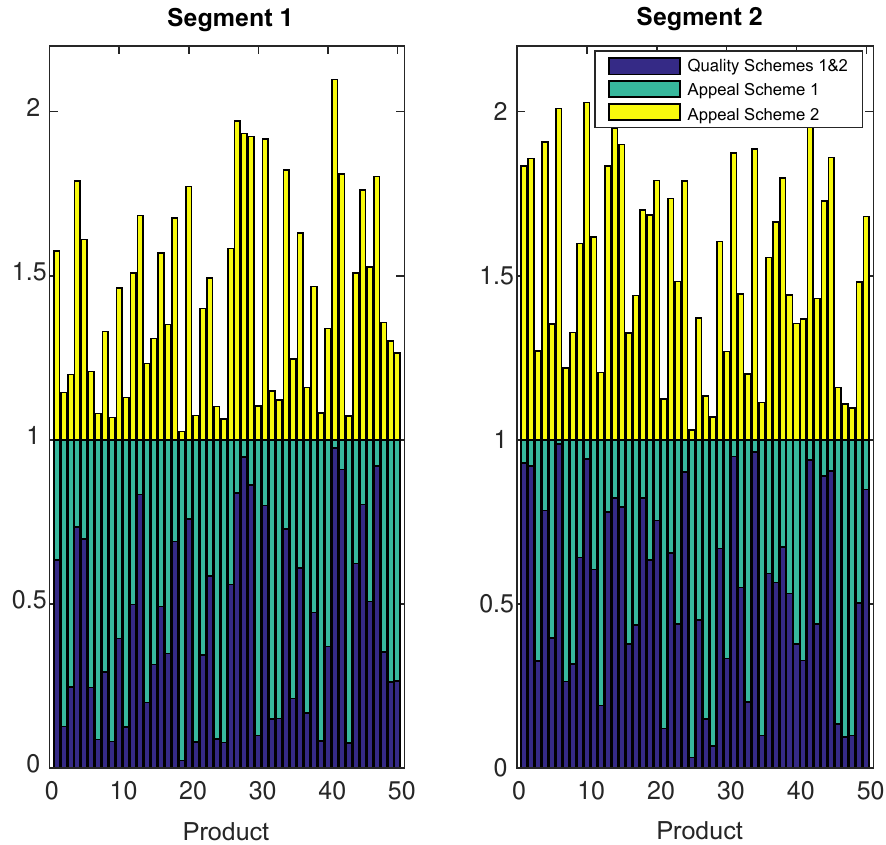}
	\caption{The quality $q_i$ (blue) and appeal $A_i$ (green and yellow) of product $i$ for settings $1$ and $2$ for both classes of consumers. The settings only differ in the appeal of items, and not in the quality of items. In Setting 1, the appeals are negatively correlated with quality, so that the sum between them is always 1. In Setting 2, the appeal is correlated to the quality with a small noise.}
	\label{fig:setting12}
\end{figure}

\begin{figure}[t]
	\centering
	\includegraphics[trim = 40mm 95mm 40mm 90mm, clip, width=12cm]{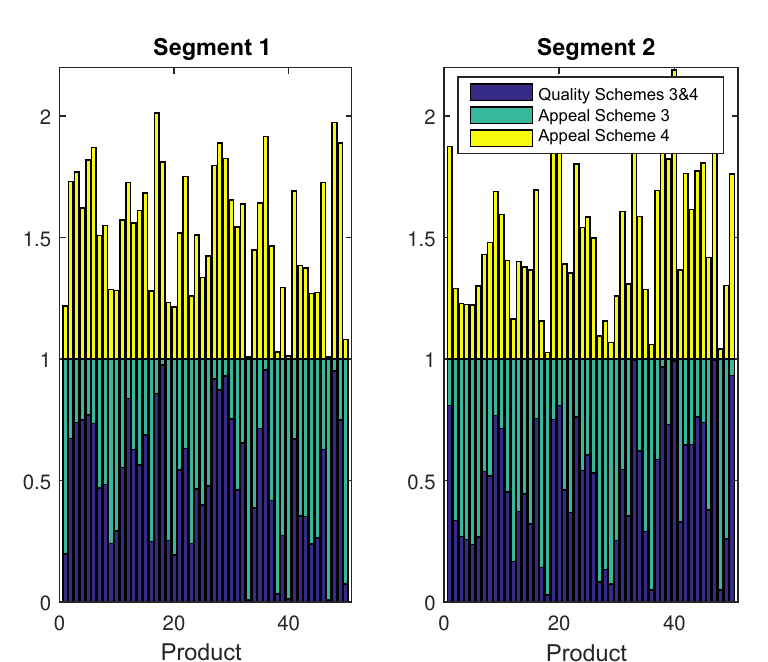}
	\caption{The quality $q_i$ (blue) and appeal $A_i$ (green and yellow) of product $i$ in settings $3$ and $4$ for both classes of consumers. The settings only differ in the appeal of items, and not in the quality of items. In Setting 3, the appeals are negatively correlated with quality, so that the sum between them is always 1. In Setting 4, the appeal is correlated to the quality with a small noise. All appeals were then multiplied by a factor of $200$ to use in the model.}
	\label{fig:setting34}
\end{figure}

To highlight and complement the theoretical results, we consider four
different schemes, the schemes share the following characteristics:
They have two customer classes with the same weight and they use 50
products.  They differ in how the values for the item appeals and
qualities are chosen. The schemes are depicted visually in Figures
\ref{fig:setting12} and \ref{fig:setting34} and were obtained as
follows:

\begin{enumerate}
	
	\item Scheme 1: The product qualities for each consumer segment were chosen
	randomly with a standard uniform distribution ($q_{i,1}$ and $q_{i,2}$ are independent for all $i\in[N]$). Appeals were
	negatively correlated with quality, i.e., $a_{i,k} = 1 - q_{i,k}$ for all $i\in[N]$.
	
	\item Scheme 2: Product qualities are similar to Scheme 1. Appeal
	vectors are now correlated with the quality vectors. More precisely, the
	appeal vector for each consumer segment was set to 0.8 times the quality plus a
	random uniform vector between -0.4 and 0.4, i.e., $a_{i,k}=q_{i,k}
	(0.8 + 0.4*\epsilon_i)$ for all $i\in[N]$, where $\epsilon_i$ is a standard uniform random variable.
	
	\item Scheme 3: The product quality for segment 1 is a random vector,
	while $q_{i,2} = 1 - q_{i,1} + 0.01*\epsilon_i$ for all $i\in[N]$, where $\epsilon_i$ is a standard uniform random variable. Appeals are negatively correlated with quality, i.e., $a_{i,k} = 1 - q_{i,k}$ for all $i\in[N]$.
	
	\item Scheme 4: The product qualities are the same as in Scheme 3 but
	the appeals are correlated with qualities, $a_{i,k}=q_{i,k} (0.8 +
	0.4 \ rand(1,50))$ for all $i\in[N]$.
\end{enumerate}

\noindent
Observe that, in Schemes 3 and 4, customers in the two classes
associate fundamentally different qualities with the products. For the simulations, the appeals vector were multiplied by a factor of $200$.

\paragraph{The Policies}

The simulations compare the average and segmented quality rankings
with and without the popularity signal. We use the following
notations:
\begin{itemize}
	\item {\tt SQSSI}: Segmented quality ranking with segmented popularity signal;
	\item {\tt SQNSI}: Segmented quality ranking without popularity signal;
	\item {\tt AQGSI}: Average quality ranking with global popularity signal;
	\item {\tt AQNSI}: Average quality ranking without popularity signal.
\end{itemize}

\subsection{Market Efficiency}

\begin{figure}[t]
	\centering
	\subfigure{
		\includegraphics[trim = 10mm 57mm 10mm 75mm, clip, width=7cm]{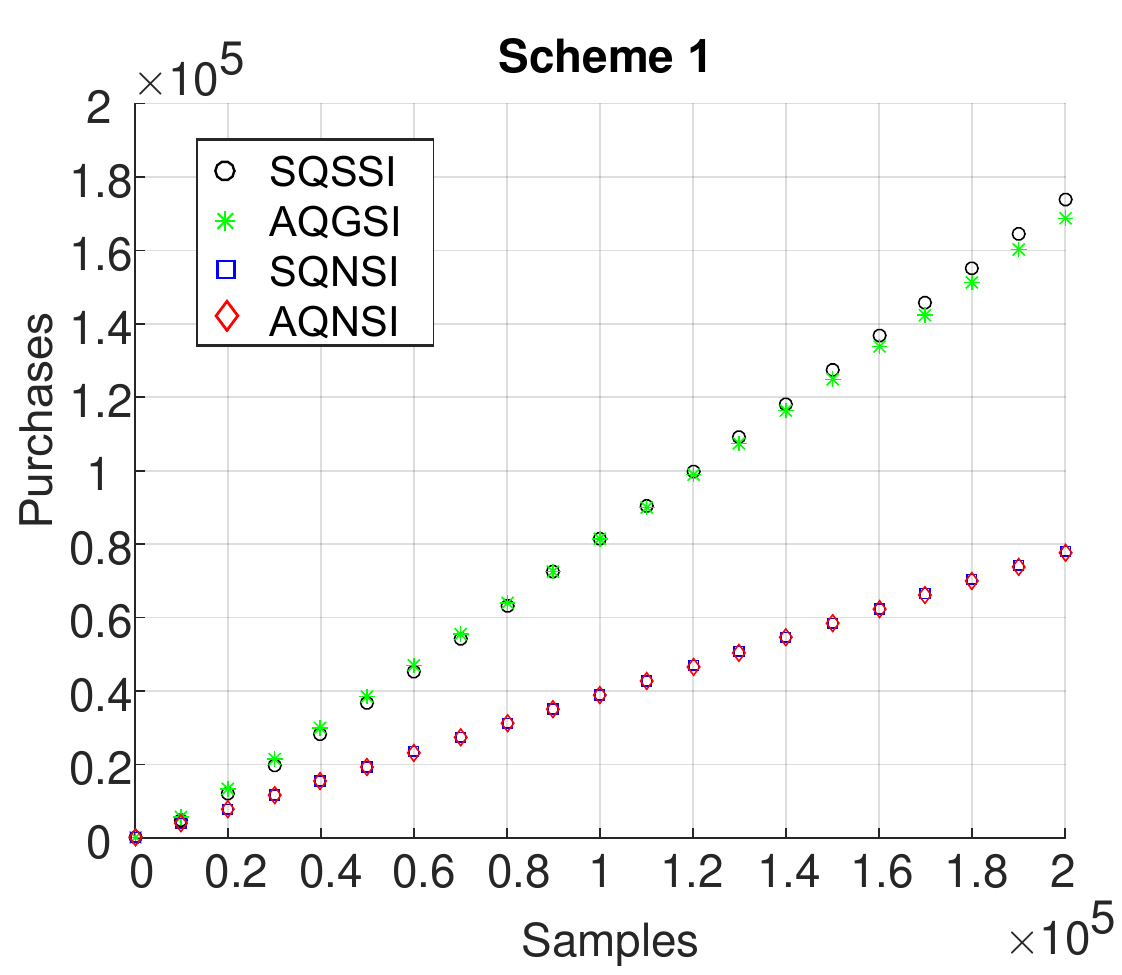}}
	\hspace{0.5cm}
	\subfigure{
		\includegraphics[trim = 10mm 57mm 10mm 75mm, clip, width=7cm]{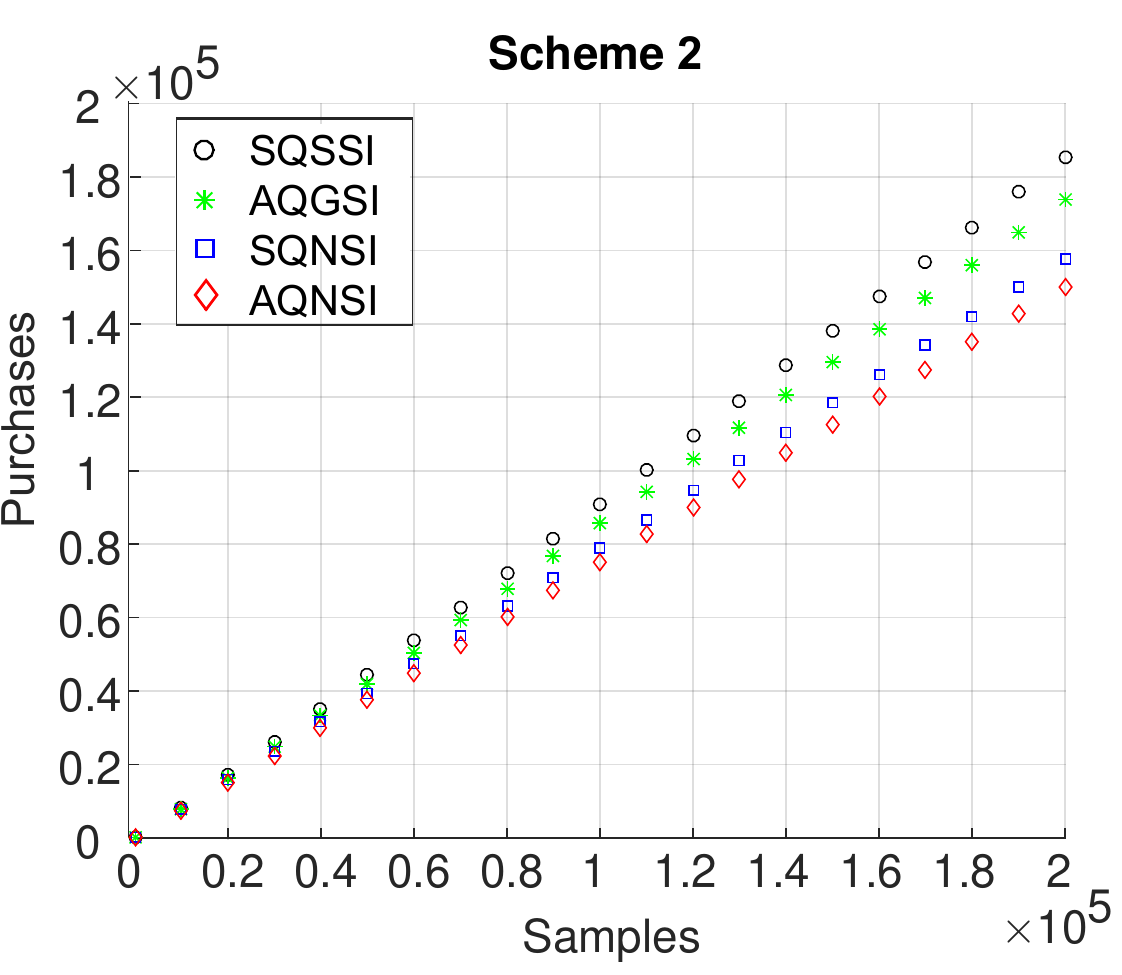}}
	\caption{The Number of Purchases over Time for the Various Rankings. The x-axis represents the number of items tried and the y-axis represents
		the average number of purchases over all experiments. The left figure depicts the results for Scheme 1 and the right figure for Scheme 2. }
	\label{fig:me12}
\end{figure}
\begin{figure}[t]
	\centering
	\includegraphics[trim = 8mm 57mm 8mm 75mm, clip, width=7cm]{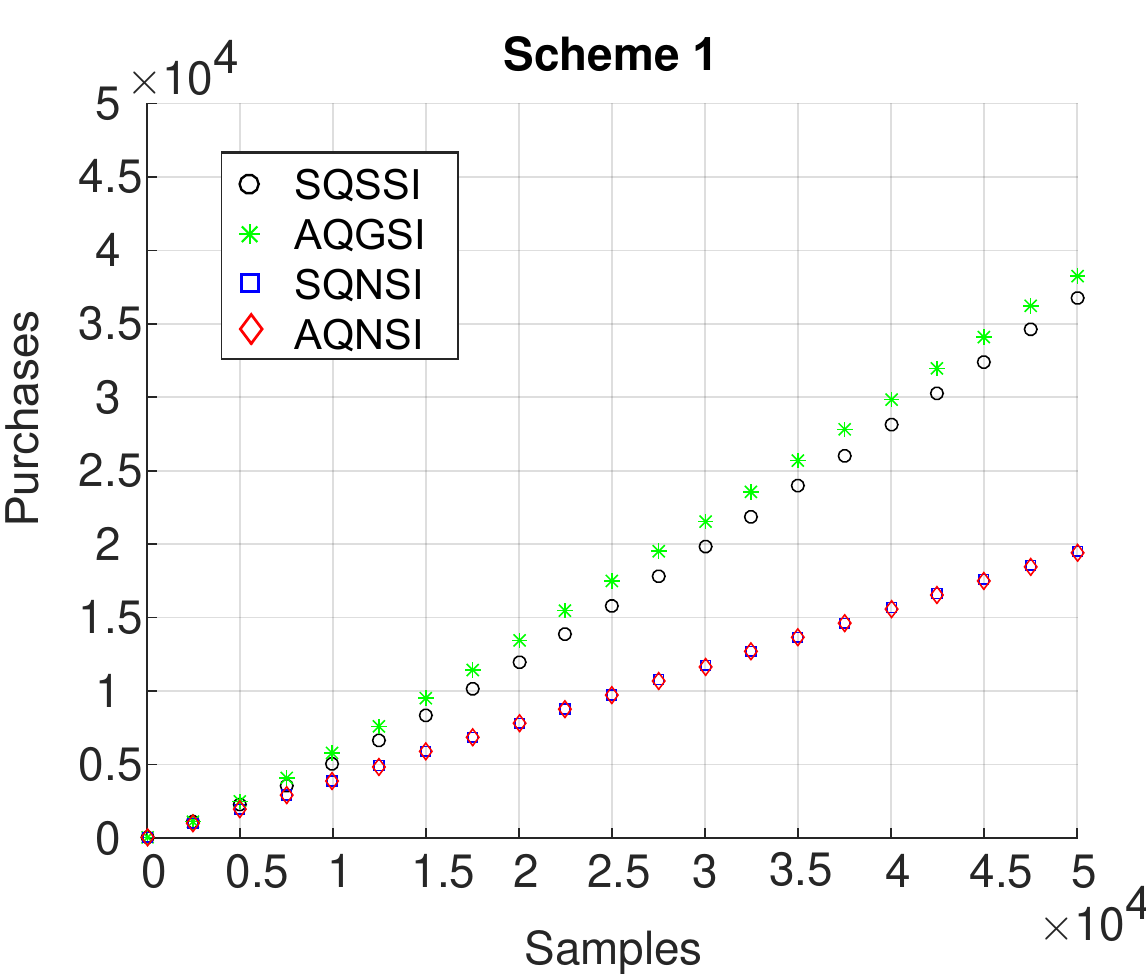}
	\caption{The number of purchases over time for the various rankings. The x-axis represents the number of items tried and the y-axis represents
		the average number of purchases over all experiments. The figure depicts the results for Scheme 1 in the early part of the simulation.}
	\label{fig:me12inset}
\end{figure}

Figure \ref{fig:me12} depicts the results for Schemes 1 and 2. For
Scheme 1, the popularity signal is beneficial for both the segmented
and average quality rankings. SQSSI is the most efficient ranking
policy. It is also interesting to observe that AQGSI outperforms
SQSSI early on before being overtaken as highlighted in Figure
\ref{fig:me12inset}. Scheme 2 exhibits similar results but the benefit of the popularity signal is lower.

\begin{figure}[t]
	\centering
	\subfigure{
		\includegraphics[trim = 10mm 57mm 10mm 75mm, clip, width=7cm]{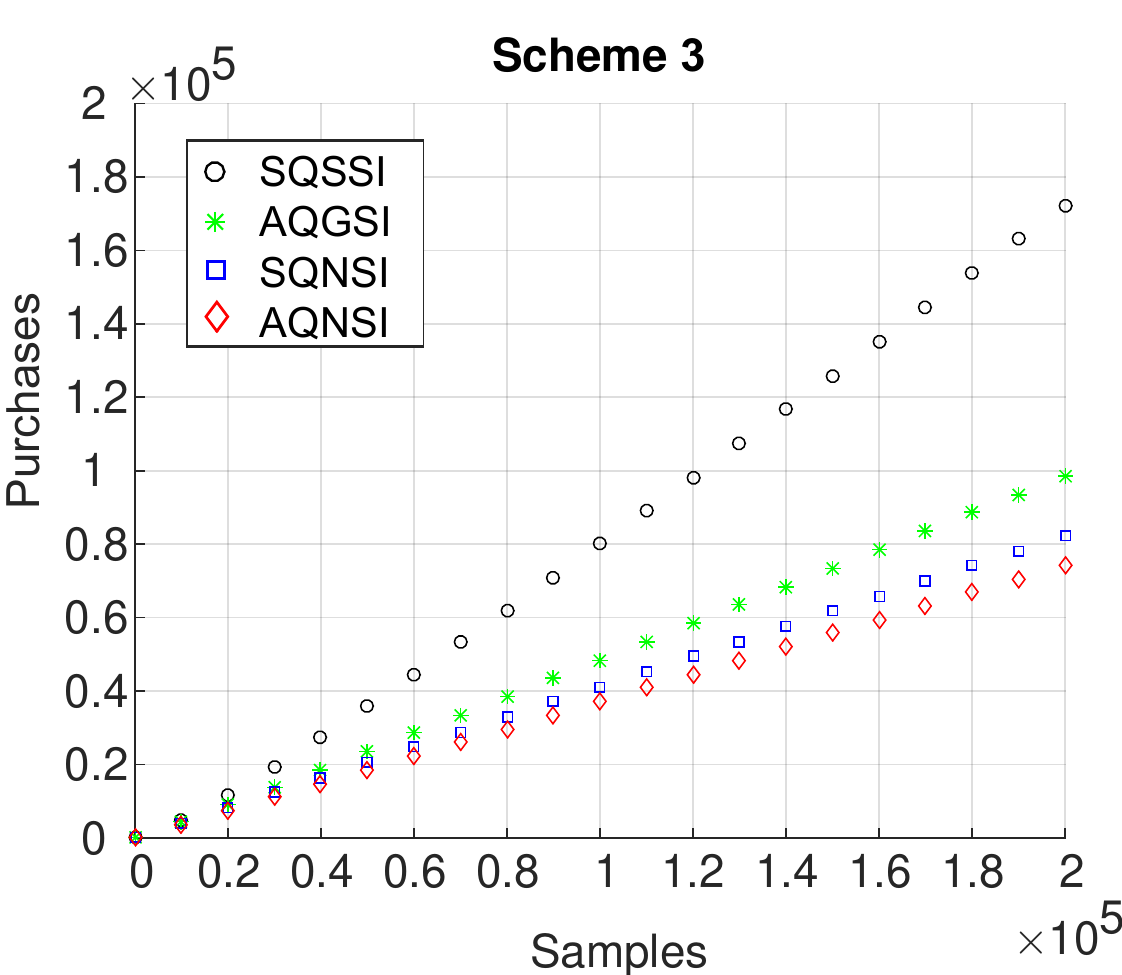}}
	\hspace{0.5cm}
	\subfigure{
		\includegraphics[trim = 10mm 57mm 10mm 75mm, clip, width=7cm]{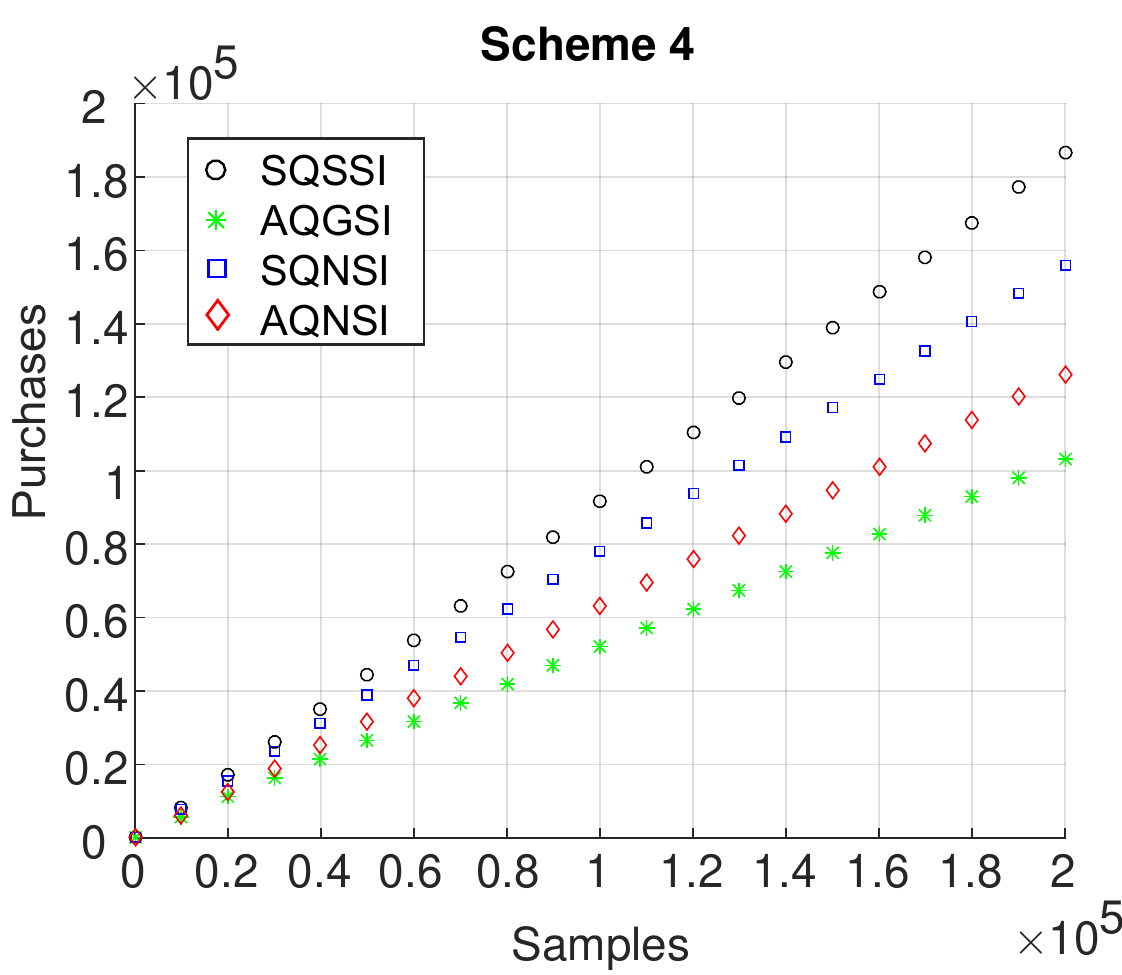}}
	\caption{The number of purchases over Time for the various rankings. The x-axis represents the number of items tried and the y-axis represents
		the average number of purchases over all experiments. The left figure depicts the results for Scheme 3 and the right figure for Scheme 4. }
	\label{fig:me34}
\end{figure}

Figure \ref{fig:me34} depicts the results for Schemes 3 and 4 and they
are particularly interesting. Recall that, in Schemes 3 and 4, the two
classes of customers have opposite preferences in terms of product
qualities. For Scheme 3, the popularity signal is again beneficial for the segmented and average
quality rankings. SQSSI is again the best ranking policy, and particularly is almost twice as efficient than AQGSI, nicely illustrating Theorem \ref{thm_bound}, since the
improvement is close to the best possible ratio. Once again, AQNSI
performs the worst. For Scheme 4, SQSSI is again the best ranking policy but
the second best policy is SQNSI, the segmented quality ranking with no
popularity signal. The worst policy is AQGSI, providing a compelling
illustration of Theorem \ref{bound-nsi}: The popularity signal may be
detrimental to the average quality ranking.

These results can be summarized as follows:
\begin{enumerate}
	\item SQSSI (segmentation with the popularity signal) is clearly the best
	policy and it dominates all other policies. Market segmentation with
	the popularity signal is very effective in these trial-offer markets.
	
	
	\item The global popularity signal may be beneficial or detrimental to the
	average quality ranking. It is detrimental when the market has
	customers with very different product preferences.
\end{enumerate}

\subsection{Purchase Profiles}

\begin{figure}[!th]
	\centering
	\subfigure{
		\includegraphics[trim = 5mm 58mm 15mm 68mm, clip, width=7cm]{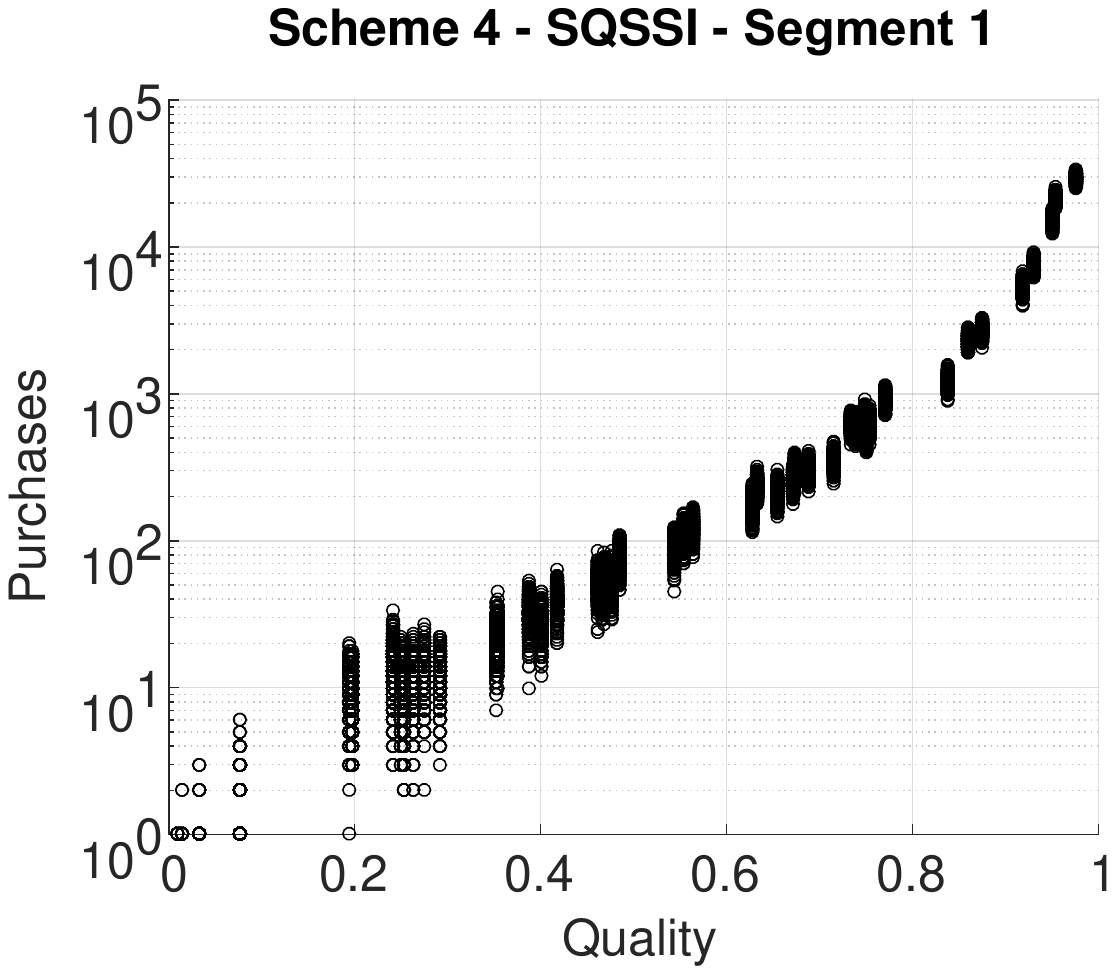}}
	\hspace{0.5cm}
	\subfigure{
		\includegraphics[trim = 5mm 58mm 15mm 68mm, clip, width=7cm]{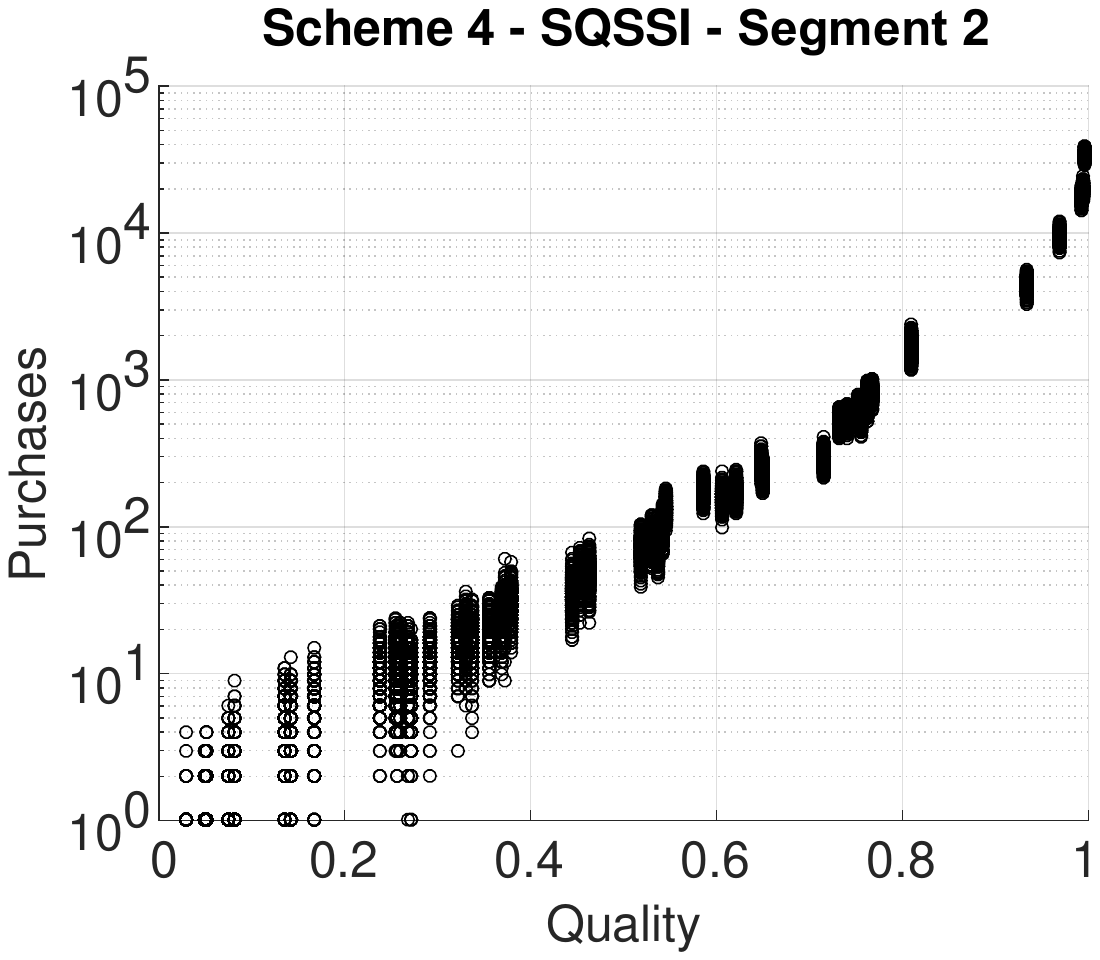}}
	\caption{The purchase profiles of SQSSI on Scheme 4 for consumer segments 1 (left) and 2 (right).}
	\label{fig:SQSSI3-12unpred}
\end{figure}

\begin{figure}[!th]
	\centering
	\subfigure{
		\includegraphics[trim = 5mm 59mm 15mm 68mm, clip, width=7cm]{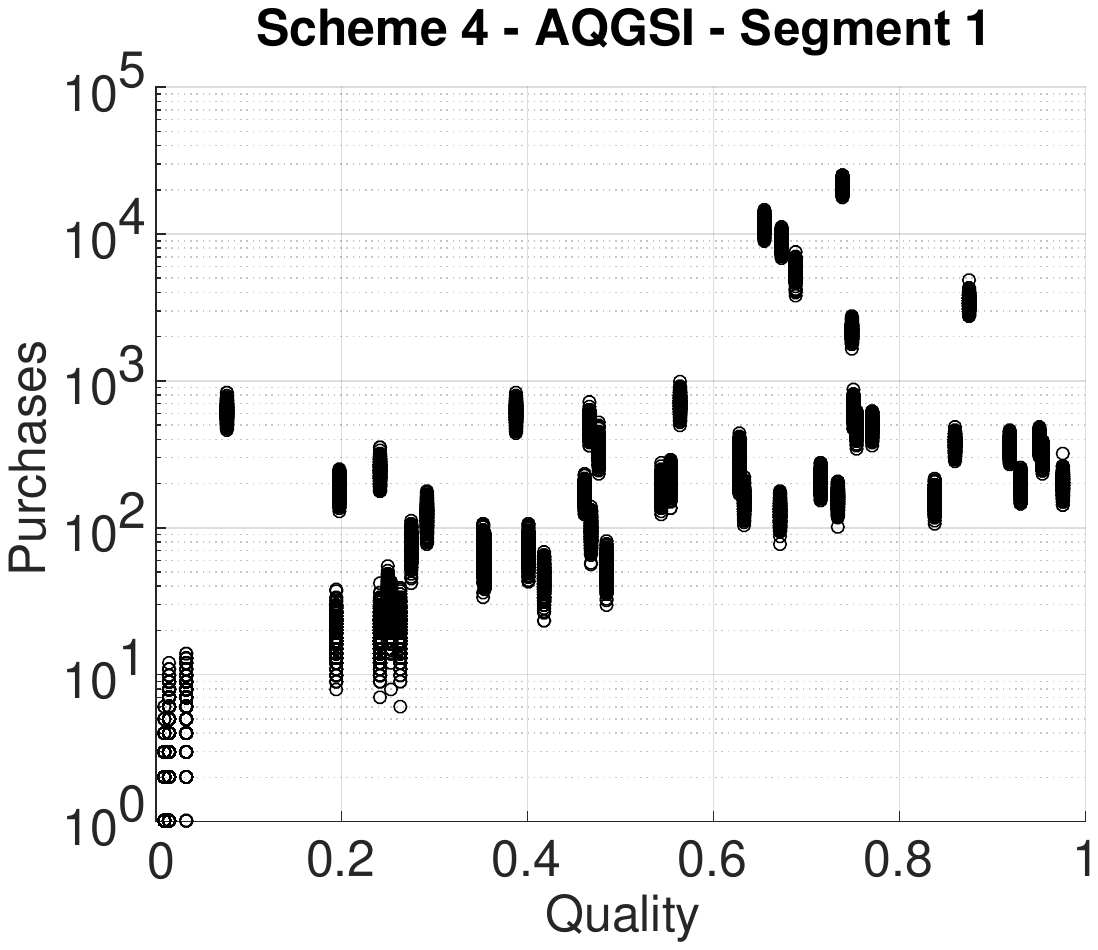}}
	\hspace{0.5cm}
	\subfigure{
		\includegraphics[trim = 5mm 59mm 15mm 68mm, clip, width=7cm]{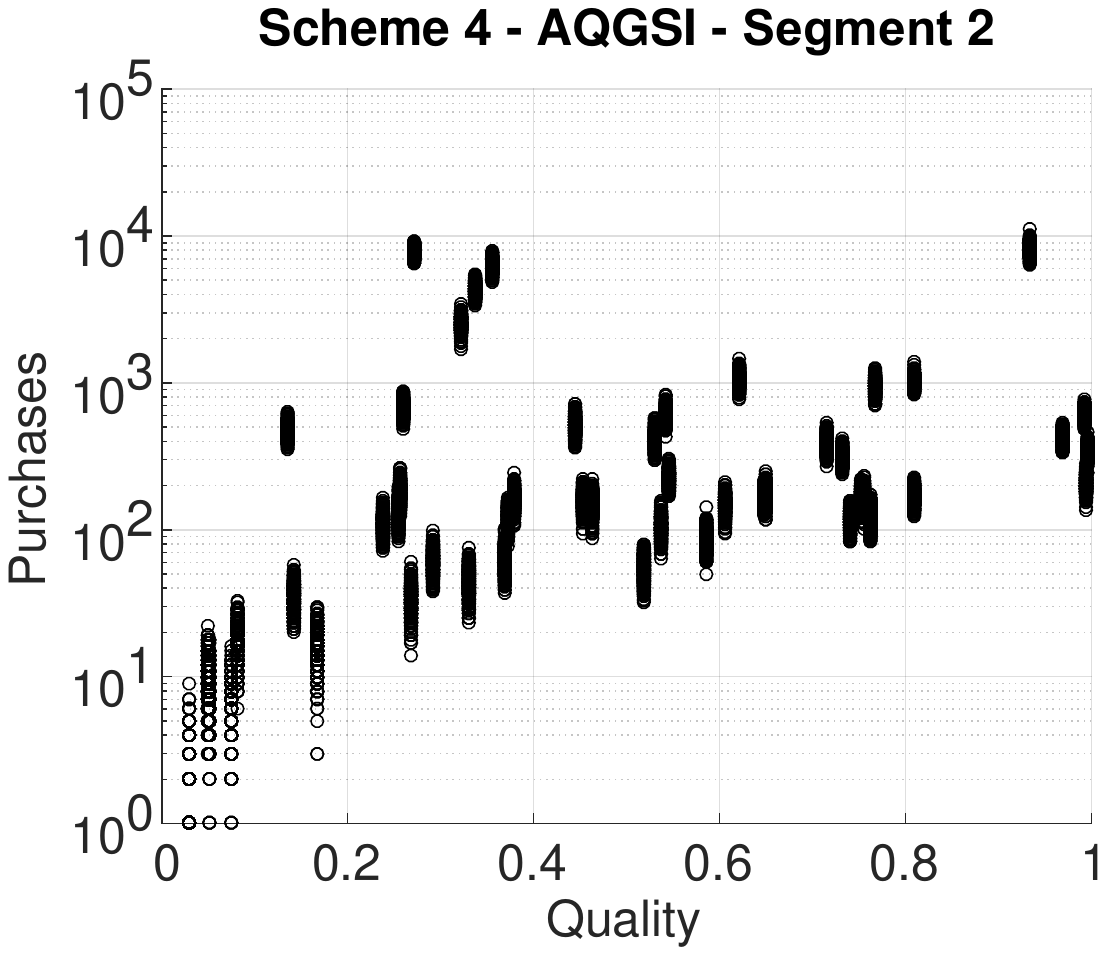}}
	\caption{The purchase profiles of AQGSI on Scheme 4 for consumer segments 1 (left) and 2 (right).}
	\label{fig:AQGSI3-12unpred}
\end{figure}

\begin{figure}[!th]
	\centering
	\subfigure{
		\includegraphics[trim = 10mm 2.2in 9mm 3in, clip, width=7cm]{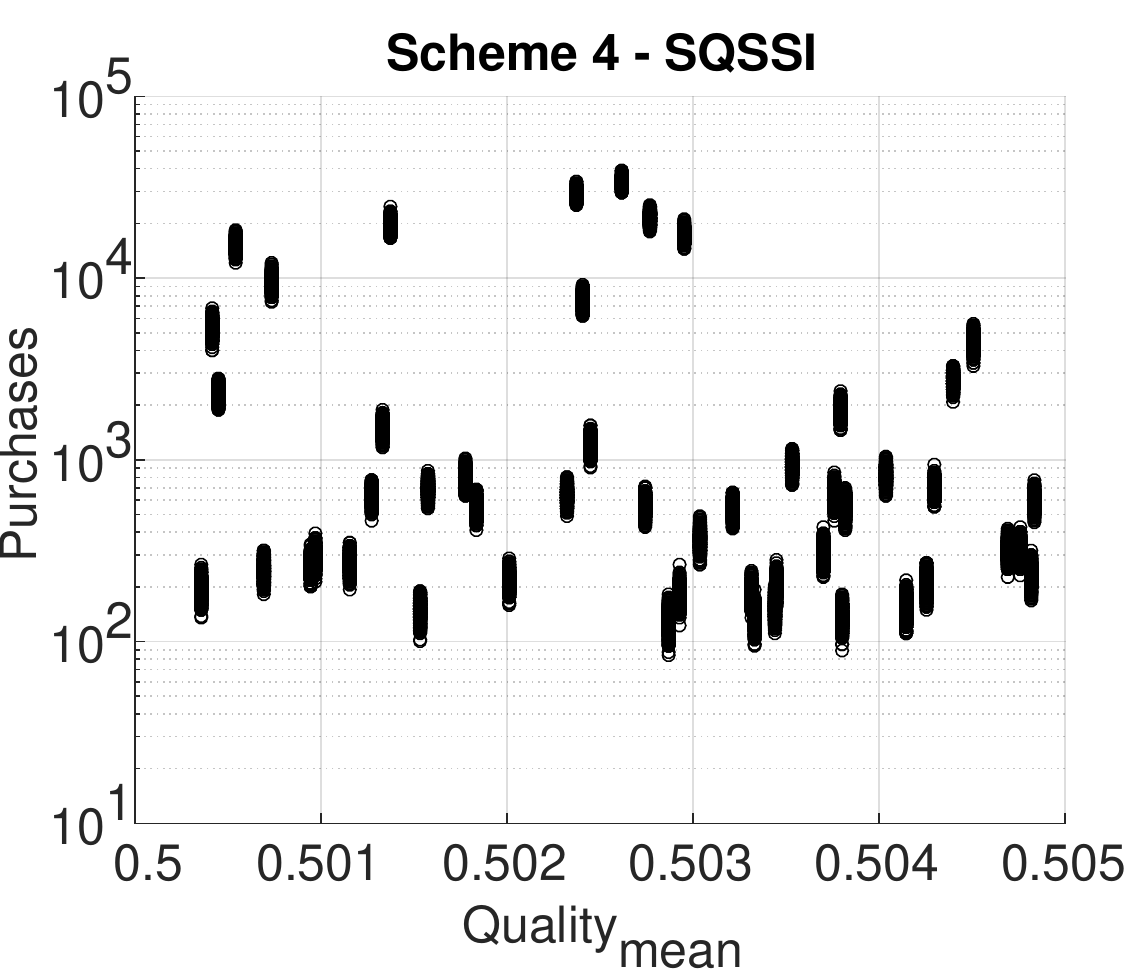}}
	\hspace{0.5cm}
	\subfigure{
		\includegraphics[trim = 10mm 2.2in 9mm 3in, clip, width=7cm]{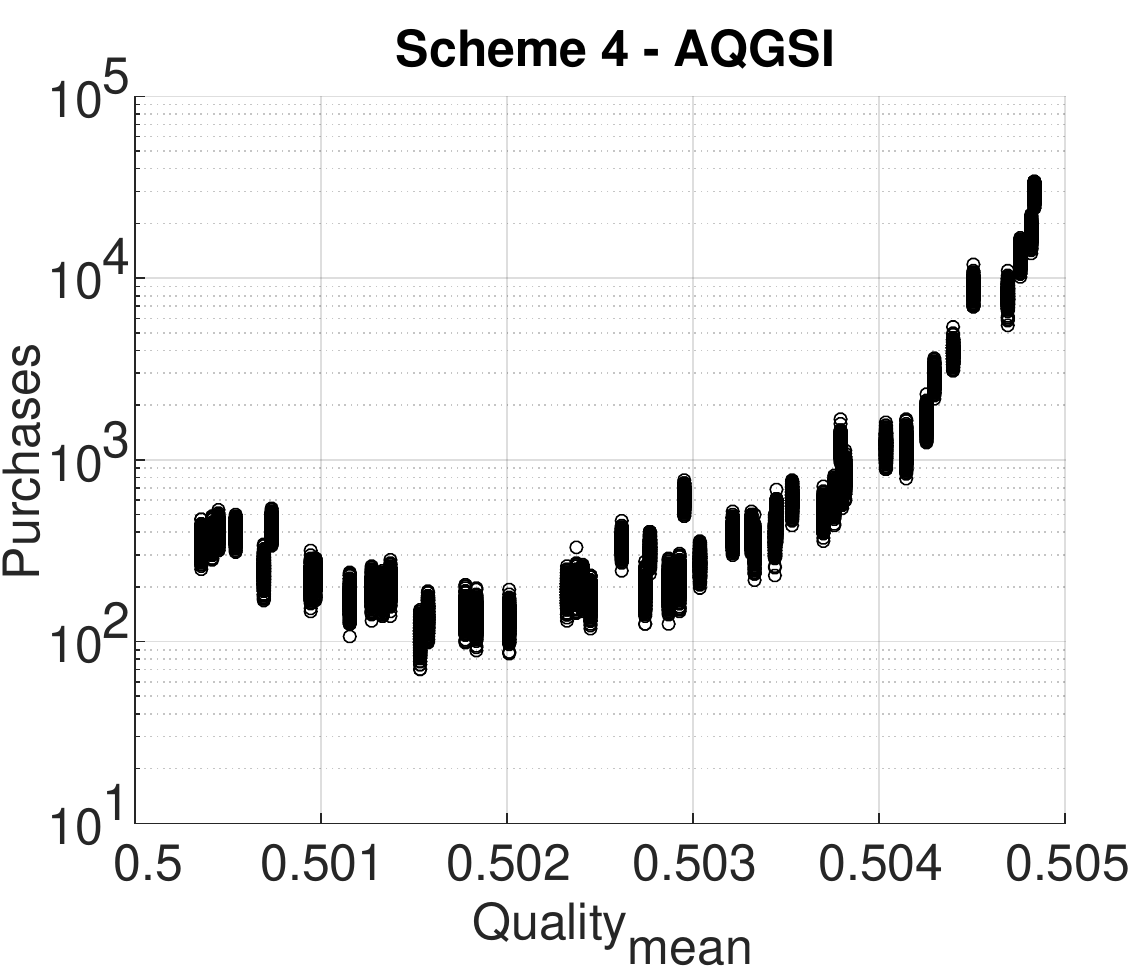}}
	\caption{The purchase profiles of SQSSI and AQGSI on Scheme 4 for both segments of customers.}
	\label{fig:unpred3}
\end{figure}

%
%

We now illustrate the customer and market behaviors for the SQSSI and
AQGSI rankings, which exhibit some significant differences. For Scheme
4, the results are presented in Figures \ref{fig:SQSSI3-12unpred},
\ref{fig:AQGSI3-12unpred}, and \ref{fig:unpred3}. Figure
\ref{fig:SQSSI3-12unpred} depicts separately the purchase profiles of customers
of segments 1 and 2 for policy SQSSI. The products are sorted by
increasing quality for each segment: i.e., the products of highest
quality for customers of segment 1 (resp. segment 2) is in the rightmost
position in the left (resp. right) picture. Since the market is
segmented, the results are not surprising and consistent with past
results: The number of purchases is strongly correlated with
quality. Figure \ref{fig:AQGSI3-12unpred} is more interesting and
depicts the same information for policy AQGSI. Here the number of
purchases is no longer correlated with quality for a specific customer
segment. Figure \ref{fig:unpred3} compares SQSSI and AQGSI over all
customers and the products are sorted by average quality. The figure
highlights a fundamental difference in market behavior between the two
policies, with very different products emerging as the ``best
sellers''.


\begin{figure}[!th]
	\centering
	\subfigure{
		\includegraphics[trim = 6mm 2.2in 10mm 3in, clip, width=7cm]{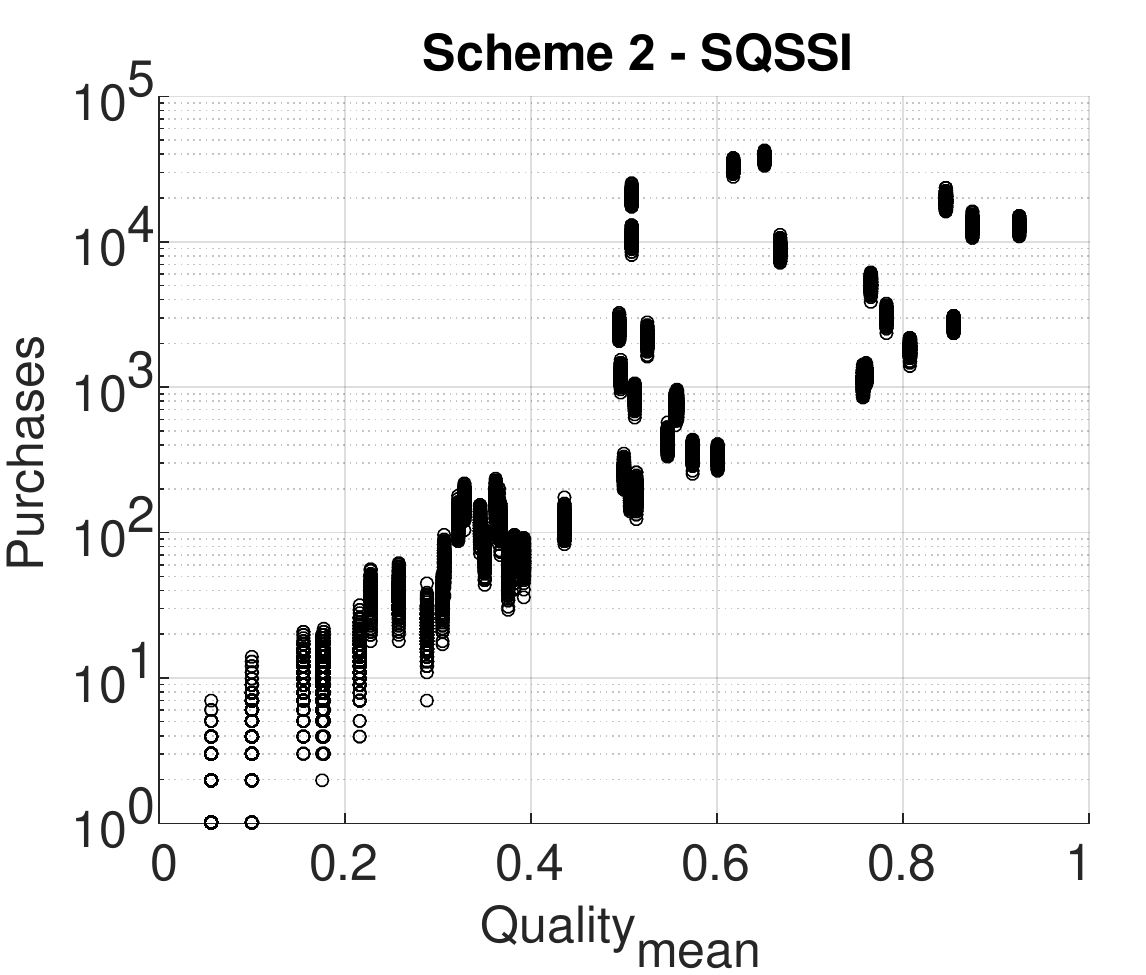}}
	\hspace{0.5cm}
	\subfigure{
		\includegraphics[trim = 6mm 2.2in 10mm 3in, clip, width=7cm]{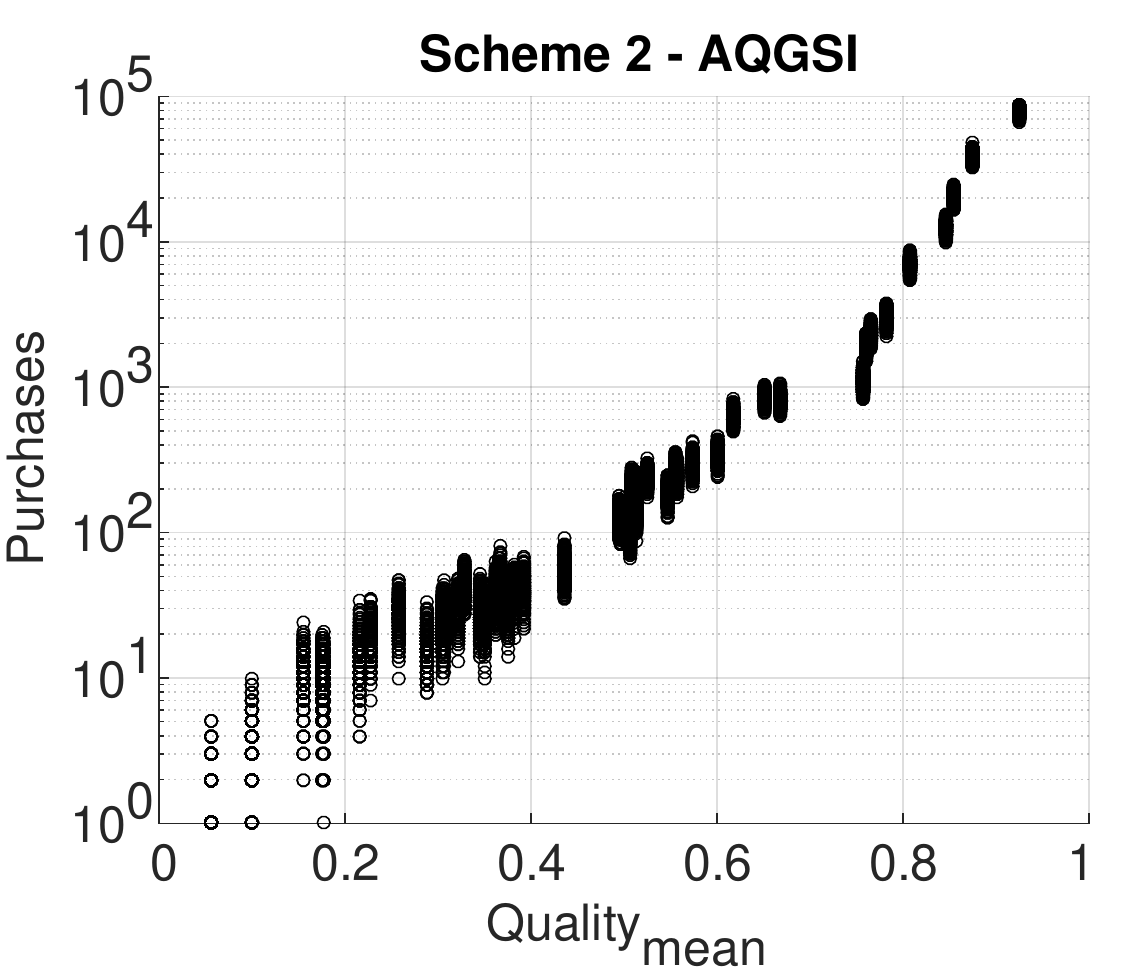}}
	\caption{The purchase profiles of SQSSI and AQGSI on Scheme 2 for both Classes of Customers.}
	\label{fig:unpred1}
\end{figure}

\begin{figure}[!th]
	\centering
	\subfigure{
		\includegraphics[trim = 6mm 2.2in 15mm 3in, clip, width=7cm]{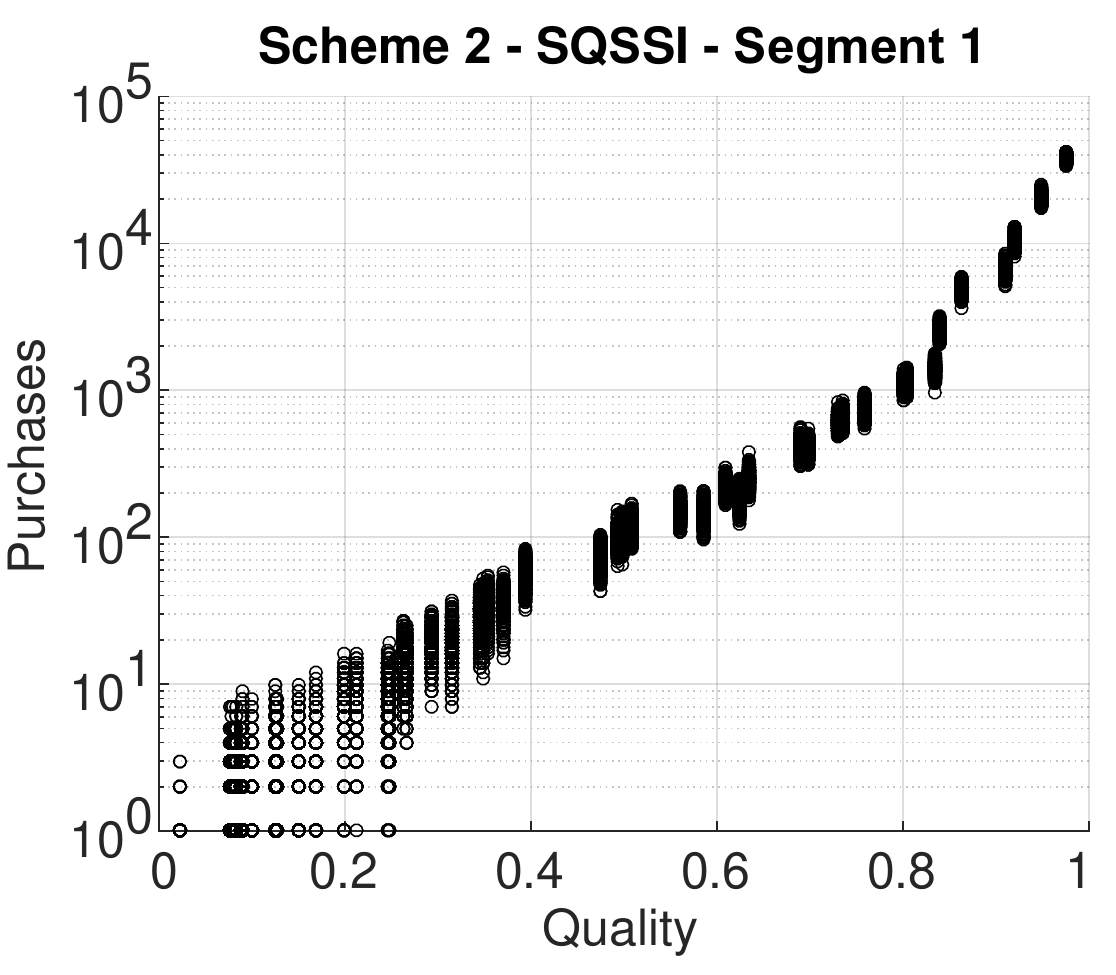}}
	\hspace{0.5cm}
	\subfigure{
		\includegraphics[trim = 6mm 2.2in 15mm 3in, clip, width=7cm]{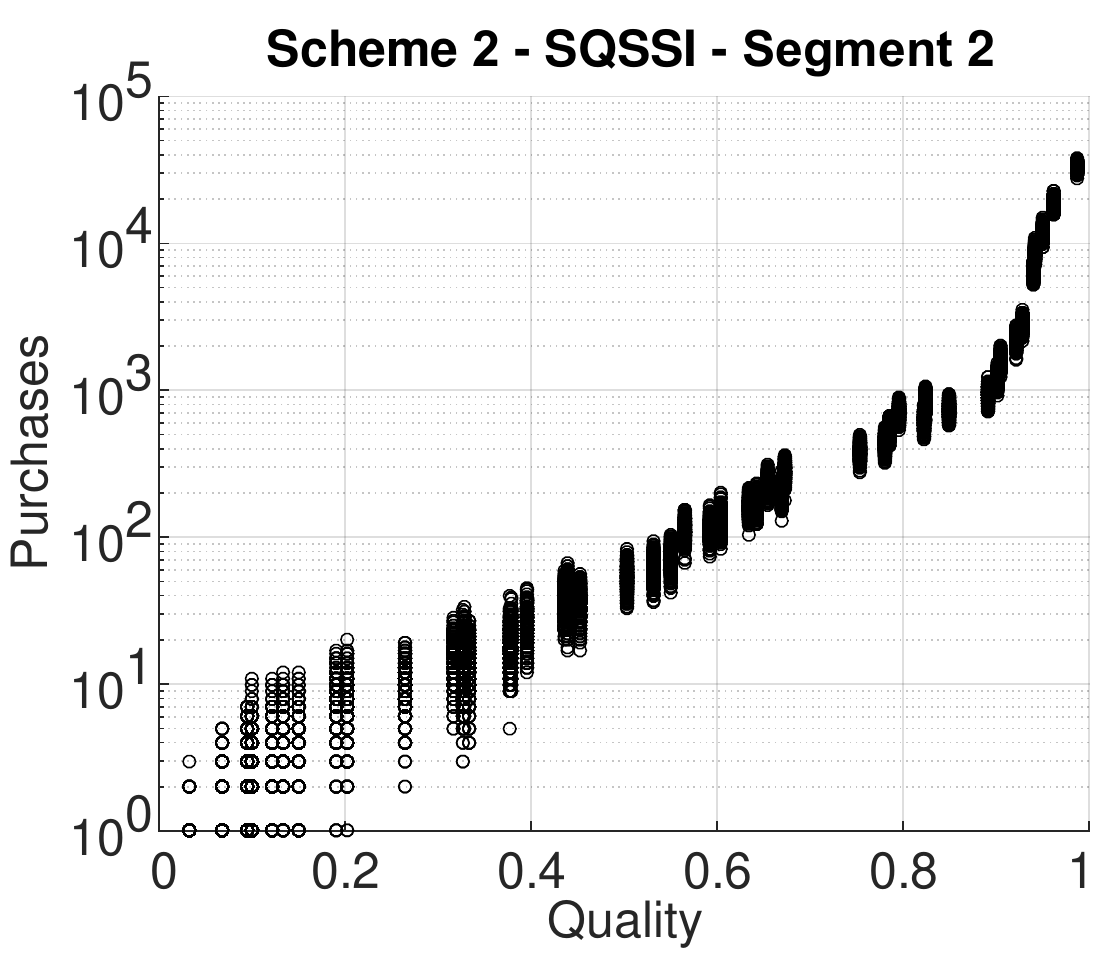}}
	\caption{The purchase profiles of SQSSI on Scheme 2 for consumer segments 1 (left) and 2 (right).}
	\label{fig:SQSSI1-12unpred}
\end{figure}

Schemes 3 and 4 feature customer classes with opposite preferences. It
is thus interesting to report the results on Scheme 2 where the
product qualities were generated independently for the two
classes. Figure \ref{fig:unpred1} depicts these results. We already
know from Figure \ref{fig:me12} that policy SQSSI outperforms AQGSI
but it is interesting to see how different the market behaves under
these two policies. For AQGSI, as expected, the products of best
average quality receives the most purchases: Asymptotically the market
goes to a monopoly for that product. For SQSSI, the purchases at this
stage of the market are distributed through a larger number of
products, each of which have fewer purchases. Asymptotically, the
market will go to a monopoly for two products (one for segment 1 and one
for segment 2) but the popularity signal is weaker for SQSSI since it is
spread across the two classes. It is interesting to observe that the
segmentation policy SQSSI is still more efficient than policy AQGSI
despite this weaker popularity signal. Figure
\ref{fig:SQSSI1-12unpred} depicts the profiles for policy SQSSI and
nicely highlights that many products are receiving significant
purchases.

{\color{black} We finish this section with a set of managerial insights that can be observed from the theoretical results and the computational experiments. First, whenever consumer segment information is available and the goal is to maximize long term purchases, it is optimal to segment the consumers and to rank the products within each segment according to their quality, while using a social signal that is only shown across consumers within the same segment (Theorems \ref{bound-nsi} and \ref{AQNSIvsSQSSI_bound}). The first theorem shows that a segmented quality ranking with social influence within each segment (SQSSI) always outperforms the unique average quality ranking with the same social signal across all consumer segments (AQGSI). The second theorem shows that using social influence with a segmented quality ranking (SQSSI) always outperforms doing a segmented quality ranking without social influence (SQNSI). In the computational experiments we observe similar results. However, the platform owner may be better off in the short term using a unique average quality ranking with a global social signal (AQGSI) (Scheme 1, Figure \ref{fig:me12inset}). This is because the aggregate ranking AQGSI enhances the social signal instead of distributing it across different segments as seen with the segmented ranking SQSSI. But in the long run the segmented ranking is optimal whereas the aggregate ranking may lead to a monopoly of suboptimal products. Finally, when consumer segment information is not available, it is not always beneficial to show social influence. This is proved in Theorem \ref{bound-nsi} and nicely illustrated with the experiments, where the average quality ranking with a social signal (AQGSI) outperforms its counter-part without a social signal (AQNSI) in all schemes except for Scheme 4 (when consumer segments exhibit opposite preferences, and product appeals are positively correlated with qualities). On aggregate rankings, the platform owner should use social influence carefully, and mainly when consumer segments have similar product preferences, as otherwise social influence may confuse consumers on which products to try and reduce the rate of purchases.}
\section{Conclusions and Future Research}
\label{section:conclusion}

In this paper, we studied a trial-offer market where consumers choices
about which products to try are affected by the display of past product
purchases and by product positions. Specifically, we focused on
studying the case in which consumer choices follow a mixed multinomial
logit model (MMNL), which generalizes the multinomial logit model
proposed by \citet{krumme2012quantifying} to explain the behavior in
an online cultural market \citep{salganik2006experimental}. Unlike the
case for the Multinomial Logit, we showed that finding the best way to
rank products at every step in order to maximize the purchases is a
computationally hard problem.

The paper then studied the performance of a ranking policy (AQGSI)
which ranks the products by the average quality (in decreasing
order). Under such policy, we proved that the trial-offer market in
the long run converges predictably to a monopoly by transforming the
MMNL model into a traditional MNL model whose appeals and qualities
depends on the popularity signal at each time step but are bounded
from below and above. Unfortunately, this average quality ranking
policy is no longer guaranteed to benefit from the popularity signal
in all cases. In other words, the rate of product purchases can
sometimes be reduced as time passes due to the fact that consumers
are able to observe (and their decisions are affected by) the number
past purchases for each product. This is in sharp contrast to the case
where there is a unique consumer segment ($K=1$)
\citep{van2016aligning}.

The paper also studied a market segmentation policy for settings in
which the firm is capable of detecting the consumer segment in
advance. Under the segmentation policy, the products are presented to
each consumer according to the quality ranking for their own class. In
addition, the popularity signal displayed only aggregates the past
purchases associated to the customers of the same class. The resulting
policy (SQSSI) is optimal asymptotically in expectation and may
improve the market efficiency up to a factor $K$ over AQGSI, where $K$ is
the number of customer classes.

Computational experiments have been presented to illustrate our
theoretical results and to show the market dynamics in the short
term. These experiments, which were carried over four different
settings, showed that there are settings in which the popularity
signal is indeed detrimental to policy AQGSI and that the segmentation
strategy produces the best possible improvement predicted by the theory.

Overall, the paper shows that, in trial offer markets, the decision to
display or not the popularity signal to consumers has to be analyzed
very carefully. In markets where consumers have very different product
preferences, we showed that the display of past purchases can be
detrimental to the rate of purchases. Intuitively, this is because
consumers watching past purchases become confused about which products
to try. On the other hand, in markets where the preferences of the
different consumer classes are not very distinct, the display of
aggregated past purchases is beneficial. Nevertheless, our theoretical
(asymptotic) results and our short-term simulation results indicate
that the segmentation policy SQSSI outperforms all other policies.
Moreover, these results highlight the fact that AQGSI and SQSSI
produce very different market behavior, even in settings where the
overall market efficiency is relatively close.

The paper leaves some interesting questions for future research. The first is to study the
market-share dynamics of the different ranking policies for more
complex consumer choice models. One weakness of the current model is
that the consumer choice model only depends on the displayed vector of
past purchases and the current ranking (as well as the appeal of the
products). However, a more sophisticated choice model could
incorporate into the consumer choice, the type of past purchase
information displayed: A consumer who observes a past purchase vector
$d$ might change the behavior depending on whether she/he knows that
the vector $d$ comes from all consumer purchases or if only comes from
consumers of it own type or class. The second one is to extend the model to scenarios when the trial probabilities depend non-linearly on the past purchases, this is studied in \cite{maldonado2018popularity} for the special case of a unique consumer segment ($K=1$). Another interesting research direction is to study the firm potential incentives to hide or mis-report some reviews and how it this affects the outcomes in terms of market share dynamics and consumer welfare.

\subsubsection*{Acknowledgements}
We are very grateful to Robert Dyson (co-editor) and to three anonymous referees for their thoughtful comments and relevant suggestions that helped us improve the paper.

\bibliographystyle{apacite}
\bibliography{references}

\subsection*{Appendix: Proofs}
\begin{proof}[Proof of Theorem \ref{theorem_hardness}]
	The proof uses the 2-Class Logit problem which is known to be
	NP-hard \citep{rusmevichientong2014assortment}. The inputs to a
	2-Class Logit instance are $N$ products, two sequences $V^1 =
	(V^1_1, V^1_2,\hdots,V^1_N)$ and $V^2 = (V^2_1, V^2_2,\hdots,V^2_N)$
	with $V^1, V^2 \in \mathbb{Q}_{+}^N$, and a number $\alpha \in
	\mathbb{R}_{[0,1]}$. Each product $i$ has a revenue $r_i \in
	\mathbb{Z}_{+}$. Each sequence $V^i$ represents a realization of the
	product utilities under a multinomial logit model. Sequence $V^1$
	(resp. $V^2$) has a realization probability of $\alpha$
	(resp. $1-\alpha$). The problem consists in finding a product
	assortment $S \subseteq [N]$ maximizing the expected revenue
	$\Pi^{Logit}$, i.e.,
	\[
	\Pi^{Logit} = \max _{S \subseteq [N]} \alpha_1 \frac{\sum_{i \in S}r_iV^1_i}{1 + \sum_{i \in S} V^1_i} + (1-\alpha) \frac{\sum_{i \in S}r_iV^2_i}{1 + \sum_{i \in S} V^2_i}.
	\]
	The proof shows that, if there exists an oracle to compute the
	performance ranking for the MMNL with two classes of consumers (i.e.,
	Equation (\ref{performance_ranking-MMNL}) with $K=2$), then the 2-Class
	logit problem can be solved in polynomial time.
	
	Given an instance of the 2-Class Logit problem, the idea is to create
	$N$ different instances of the performance-ranking problem in order to
	capture the various possible assortments. The $N$ instances have a
	common core. Each of them has the same $N$ items and two classes of
	consumers (i.e., $K=2$). For each consumer segment $j \in \{0,1\}$ and
	each item $i \in [N]$, we set the appeal of item $i$ for segment $j$ to
	satisfy $a_{i,j}=V^j_i$. Similarly, for each consumer segment $j \in
	\{0,1\}$ and each item $i \in [N]$, we set the quality of $i$ for
	segment $j$ to satisfy $q_{i,j}=r_i$. Note that the quality of item $i$
	is the same for both classes.  The weights of classes $0$ and $1$ are
	$\alpha$ and $1-\alpha$ respectively. We also set $z=1$ and $t=0$
	which implies that $d^t=0$. The $N$ instances differ in the position
	visibilities. In instance $i$ ($i \in [N]$), the visibility of
	position $j \in [N]$ is:
	\[
	v_j = \left\{
	\begin{array}{ll}
	1 & \textrm{ if } j \leq i  \\
	0 & \textrm{ otherwise}.
	\end{array}\right.
	\]
	\noindent
	Let $\Pi^{PR}_i$ denote the short-term optimal value of the performance ranking
	for problem instance $i$ and let $\mathcal{S}_i$ denote the collection
	of all possible subsets of products whose size is $i$, i.e.,
	$\mathcal{S}_i= \{S \subseteq [N]: |S|=i\}$. Define
	$\Pi^{Logit}_i$ as the following optimization problem:
	\[
	\Pi^{Logit}_i = \max _{S \in \mathcal{S}_i} \alpha_1 \frac{\sum_{i \in S}r_iV^1_i}{1 + \sum_{i \in S} V^1_i} + (1-\alpha) \frac{\sum_{i \in S}r_iV^2_i}{1 + \sum_{i \in S} V^2_i}.
	\]
	It follows that
	\begin{align}
	\Pi^{Logit} = \max_{i=1,\hdots,N} \Pi^{Logit}_i. \label{equality_logit_2_class}
	\end{align}		
	\noindent
	We now show that $\Pi^{PR}_i$ is equal to $\Pi^{Logit}_i$.
	\begin{align} \label{reduction_1}
	\Pi^{PR}_i &= \max_{\sigma \in S_n} \bigg\{\sum_{c=1}^2 \Big( w_c \cdot \sum_{\ell=1}^N (p_{i}(\sigma,0) \cdot q_{\ell,c}) \Big) \bigg\}, \\
	&= \max_{\sigma \in S_n} \bigg\{ \sum_{c=1}^2 \Big( w_c  \cdot \sum_{\ell=1}^N \Big(\frac{v_{\sigma(\ell)} (a_{\ell,k})}{\sum_{j=1}^N v_{\sigma(j)} (a_{j,c})+1}  \cdot q_{\ell,k} \Big) \Big) \bigg\}.\\
	=& \max_{\sigma \in S_n} \bigg\{ \alpha  \cdot \sum_{\ell=1}^N \Big(\frac{ v_{\sigma(\ell)} V^1_{\ell} r_{\ell}} {\sum_{j=1}^N v_{\sigma(j)} V^1_j+1} \Big)  +  (1-\alpha)  \cdot \sum_{\ell=1}^N \Big(\frac{ v_{\sigma(\ell)} V^2_{\ell}r_{\ell}}{\sum_{j=1}^N v_{\sigma(j)} V^2_j+1}  \Big)  \bigg\} \label{last_step_reduction_1}\\
	=& \max_{S \in \mathcal{S}_i} \bigg\{  \alpha  \cdot \sum_{\ell \in S} \Big(\frac{ V^1_{\ell} r_{\ell}}{\sum_{j \in S} V^1_j+1} \Big)  +  (1-\alpha)  \cdot \sum_{\ell \in S} \Big(\frac{ V^2_{\ell}r_{\ell}}{\sum_{j \in S} V^2_j+1}  \Big)  \bigg\} \label{last_step_reduction_2} \\
	= & \Pi^{Logit}_i
	\end{align}		
	\noindent where the equivalence between \eqref{last_step_reduction_1} and
	\eqref{last_step_reduction_2} follows from the fact that the first $i$
	positions have visibility of 1 and the remaining ones have a
	visibility of 0 and therefore selecting a permutation $\sigma \in S_n$
	reduces to deciding which $i$ items should be assigned the top $i$
	positions. As a consequence, using \eqref{equality_logit_2_class}, we have
	\begin{align}
	\Pi^{Logit} = \max_{i=1,\hdots,N} \Pi^{Logit}_i = \max_{i=1,\hdots,N} \Pi^{PR}_i.
	\end{align}
	We have shown that, by using an oracle to solve $N$ instances of the
	performance-ranking problem, it is possible to solve the original
	2-class logit problem instance in polynomial time. Hence, the
	performance ranking is NP-hard under Turing reductions.
	
\end{proof}

\begin{proof}[Proof of Lemma \ref{qsuc}]
	The market share of item $i^*$ at any period of time $t>\hat{t}$ for
	this system would be underestimated by considering the following set
	of qualities and appeals:
	\begin{equation*}
	q_{i,new}=\left\{
	\begin{array}{l l}
	q_{i,min} & \quad \text{if $i=i^*$}\\
	q_{i,max} & \quad \text{if $i\neq i^*$}
	\end{array} \right.
	\text{and}\text{  }
	a_{i,new}=\left\{
	\begin{array}{l l}
	a_{i,min} & \quad \text{if $i=i^*$}\\
	a_{i,max} & \quad \text{if $i\neq i^*$}
	\end{array} \right..
	\end{equation*}
	If this new set of qualities satisfies that
	$v_{\sigma(i^*)}q_{i^*,new} > v_{\sigma(i)}q_{i,new}$ for all $i\in
	[N] \setminus \{i^*\}$, it follows from the convergence result in
	\citep{van2016aligning} (Theorem 4.3) that the system goes
	to a monopoly for item $i^*$. Therefore, the original system also goes
	to a monopoly for item $i^*$.
\end{proof}

\begin{proof}[Proof of Theorem \ref{monopoly_alltogether}]
	The proof first shows that the MMNL model can be reduced to a
	Multinomial Logit Model whose item appeals and qualities are
	functions of the vector of purchases at each time $t$. It then shows
	that these functions stay in the bounded range, so that it is
	possible to apply Lemma \ref{qsuc}.
	
	When the same ranking $\sigma$ and popularity signals are shown to
	all consumers, the probability that item $i$ is purchased in time
	period $t$ is given by
	\begin{align}\label{eq:thm2_main}
	P_i(\sigma,d^t) &= \sum_{k=1}^K \Big( w_k \cdot
	\Big(v_{\sigma(i)}\frac{ (a_{i,k} + d_i^t)}{\sum_{j=1}^N
		v_{\sigma(j)} (a_{j,k} + d_j^t)+z_k} \cdot q_{i,k} \Big) \Big).
	\end{align}
	
	\noindent
	By rearranging the previous expression, it comes
	\begin{align*}
	P_i(\sigma,d^t) & =  \sum_{k=1}^K \frac{w_k q_{i,k} v_{\sigma(i)}}{\sum_{j=1}^Nv_{\sigma(j)}( a_{j,k}+d_j^t)+z_k} a_{i,k}+\sum_{k=1}^K \frac{w_k\ q_{i,k} v_{\sigma(i)}}{\sum_{j=1}^N v_{\sigma(j)}( a_{j,k}+d_j^t)+z_k}d_i^t \\
	& =   \sum_{k=1}^K \frac{w_k q_{i,k} v_{\sigma(i)}}{\sum_{j=1}^Nv_{\sigma(j)}( a_{j,k}+d_j^t)+z_k} a_{i,k}+d_i^t\left(\sum_{k=1}^K \frac{w_k\ q_{i,k} v_{\sigma(i)}}{\sum_{j=1}^N v_{\sigma(j)}( a_{j,k}+d_j^t)+z_k}\right) \\
	& =   \left(\sum_{k=1}^K \frac{w_k\ q_{i,k} v_{\sigma(i)}}{\sum_{j=1}^N v_{\sigma(j)}( a_{j,k}+d_j^t)+z_k}\right)\left(\frac{\left(\sum_{k=1}^K\frac{w_kq_{i,k} v_{\sigma(i)}}{\sum_{j=1}^Nv_{\sigma(j)}(a_{j,k}+d_j^t)+z_k}a_{i,k}\right)}{\left(\sum_{k=1}^K\frac{w_kq_{i,k}v_{\sigma(i)}}{\sum_{j=1}^Nv_{\sigma(j)}(a_{j,k}+d_j^t)+z_k}\right)}+ d_i^t\right) \\
	& =   \left(\sum_{k=1}^K \frac{w_k\ q_{i,k} v_{\sigma(i)}}{\sum_{j=1}^N v_{\sigma(j)}( a_{j,k}+d_j^t)+z_k}\right)\left(\frac{v_{\sigma(i)}\left(\sum_{k=1}^K\frac{w_kq_{i,k} a_{i,k}}{\sum_{j=1}^Nv_{\sigma(j)}(a_{j,k}+d_j^t)+z_k}\right)}{v_{\sigma(i)}\left(\sum_{k=1}^K\frac{w_kq_{i,k}}{\sum_{j=1}^Nv_{\sigma(j)}(a_{j,k}+d_j^t)+z_k}\right)}+ d_i^t\right)  \\
	& =   \left(\sum_{k=1}^K \frac{w_k\ q_{i,k} v_{\sigma(i)}}{\sum_{j=1}^N
		v_{\sigma(j)}(
		a_{j,k}+d_j^t)+z_k}\right)\left(\frac{\left(\sum_{k=1}^K\frac{w_kq_{i,k}
			a_{i,k}}{\sum_{j=1}^Nv_{\sigma(j)}(a_{j,k}+d_j^t)+z_k}\right)}{\left(\sum_{k=1}^K\frac{w_kq_{i,k}}{\sum_{j=1}^Nv_{\sigma(j)}(a_{j,k}+d_j^t)+z_k}\right)}+ d_i^t\right).
	\end{align*}
	
	\noindent
	Now, for each item $i$ and each time period $t$, define the function
	\begin{equation}\label{generalized_appeal}
	\widetilde{a}_i(t) =\left(\sum_{k=1}^K\frac{w_kq_{i,k} a_{i,k}}{\sum_{j=1}^Nv_{\sigma(j)}(a_{j,k}+d_j^t)+z_k}\right)\Big/\left(\sum_{k=1}^K\frac{w_kq_{i,k}}{\sum_{j=1}^Nv_{\sigma(j)}(a_{j,k}+d_j^t)+z_k}\right)
	\end{equation}
	which depends on the total number of purchases at time $t$. Using this
	definition, we have that:
	\begin{align*}
	P_i(\sigma,d^t)  = \left(\sum_{k=1}^K \frac{w_k\ q_{i,k} v_{\sigma(i)}}{\sum_{j=1}^N v_{\sigma(j)}( a_{j,k}+d_j^t)+z_k}\right)\left( \widetilde{a}_i(t)+ d_i^t\right) = \left(\sum_{k=1}^K \frac{w_k\ q_{i,k} }{\sum_{j=1}^N v_{\sigma(j)}( a_{j,k}+d_j^t)+z_k}\right)v_{\sigma(i)}\left( \widetilde{a}_i(t)+ d_i^t\right).		
	\end{align*}
	
	\noindent
	By dividing and multiplying by $\sum_{j=1}^N  v_{\sigma(j)}(\widetilde{a}_j(t) +d_j^t)$, $P_i(\sigma,d^t)$ becomes
	\begin{align*}
	\left(\sum_{k=1}^K \frac{w_k\ q_{i,k} }{\sum_{j=1}^N v_{\sigma(j)}( a_{j,k}+d_j^t)+z_k}\right)\left(\sum_{j=1}^N  v_{\sigma(j)}(\widetilde{a}_j(t) +d_j^t)\right)\left( \frac{v_{\sigma(i)}\left(\widetilde{a}_i(t)+ d_i^t\right)}{\sum_{j=1}^N  v_{\sigma(j)}(\widetilde{a}_j(t) +d_j^t)}\right).
	\end{align*}
	
	\noindent
	Now define the following function for each item $i$ at each time period $t$:
	\begin{equation}\label{generalized_q}
	\widetilde{q}_i(t) = \left(\sum_{k=1}^K \frac{w_k\ q_{i,k} }{\sum_{j=1}^N v_{\sigma(j)}( a_{j,k}+d_j^t)+z_k}\right)\left(\sum_{j=1}^N  v_{\sigma(j)}(\widetilde{a}_j(t) +d_j^t)\right).
	\end{equation}
	The probability of purchasing product $i$ in the next iteration
	becomes:
	\begin{equation*}
	P_i(\sigma,d^t)= \left(\frac{v_{\sigma(i)}(\widetilde{a}_i(t)+ d_i^t)}{\sum_{j=1}^Nv_{\sigma(j)}(\widetilde{a}_j(t)+ d_j^t)}\right) \widetilde{q}_i(t).
	\end{equation*}
	This is almost a multinomial logit model, except that the quality and
	appeal vectors that depend on time. When the number of iterations $t$
	tends to infinity, the total number of purchases $\sum_{j=1}^Nd_j^t$
	also goes to infinity. Moreover, as $t$ goes to infinity, the
	generalized appeal ($\widetilde{a}_i(t)$) and quality ($\widetilde{q}_i(t)$)  for every item converges to
	
	\begin{equation}\label{quality-appeal-limits}
	\bar{a}_i \doteq \lim_{t \to\infty}\widetilde{a}_i(t) = \frac{\sum_{k=1}^K w_ka_{i,k}q_{i,k}}{\sum_{k=1}^K w_kq_{i,k}} \quad \text{and} \quad \bar{q}_i\doteq\lim_{t \to\infty}\widetilde{q}_i(t) = \sum_{k=1}^K w_kq_{i,k}.
	\end{equation}		
	\noindent
	In addition, observe that $\widetilde{Q}_{i}(t)\doteq v_{\sigma(i)}
	\widetilde{q}_i(t)$ also converges when $t$ goes to infinity:
	\begin{equation*}
	\bar{Q}_i\doteq\lim_{t \to\infty} v_{\sigma(i)} \widetilde{q}_i(t) = v_{\sigma(i)} \bar{q}_i
	\end{equation*}
	
	\noindent
	The tie-breaking condition (Equation \ref{cond1}) guarantees that
	there exists only one item $i^*$ such that $i^* = \argmax_{i \in [N]}
	\bar{Q}_i$. Let $i^{**}$ be the item with the second highest value
	$\bar{Q}_i$, i.e., $\bar{Q}_{i^{**}} \geq \bar{Q}_{j}$ for all $j \in
	[N], j \neq i^*$. Consider now the following difference $\Delta
	\bar{Q}=\bar{Q}_{i^*}-\bar{Q}_{i^{**}}$. Equation
	\eqref{generalized_appeal} can be seen as a weighted average on $k$
	for $a_{i,k}$ and hence	
	\begin{equation}
	\label{boundsAk}
	\min_{1\leq k\leq K}a_{i,k} \leq \widetilde{a}_i(t)  \leq \max_{1\leq k\leq K}a_{i,k} \quad \forall i\in[N], t \in \mathbb{N}
	\end{equation}	
	\noindent
	Moreover, by applying this result to Equation \eqref{generalized_q}, we obtain the following bounds for $\widetilde{q}_i$:		
	\begin{align*}
	\label{boundsqk}
	\widetilde{q}_i(t) & \geq \sum_{k=1}^K \frac{w_kq_{i,k}}{\sum_{j=1}^Nv_{\sigma(j)}(\max_{1\leq k\leq K}a_{j,k}+d_j^t)+z_k}\left(\sum_{j=1}^N  v_{\sigma(j)}(-\max_{1\leq k\leq K}a_{j,k}+\max_{1\leq k\leq K}a_{j,k} +d_j^t)-z_k+z_k\right)\\
	& = \sum_{k=1}^K \frac{w_kq_{i,k}\left(\sum_{j=1}^N  v_{\sigma(j)}(\max_{1\leq k\leq K}a_{j,k} +d_j^t)+z_k\right)-w_kq_{i,k}\sum_{j=1}^N v_{\sigma(j)}\max_{1\leq k\leq K}a_{j,k}-w_kq_{i,k}z_k}{\sum_{j=1}^Nv_{\sigma(j)}(\max_{1\leq k\leq K}a_{j,k}+d_j^t)+z_k}\\
	& \geq \left(1-\frac{\sum_{j=1}^N v_{\sigma(j)} \max_{1\leq k\leq K}a_{j,k}+\max_{1\leq k \leq K}z_k}{\sum_{j=1}^N v_{\sigma(j)}d_j^t}\right)\sum_{k=1}^K w_kq_{i,k} \\
	& = \left(1-\frac{\sum_{j=1}^N v_{\sigma(j)} \max_{1\leq k\leq K}a_{j,k}+\max_{1\leq k \leq K}z_k}{\sum_{j=1}^N v_{\sigma(j)}d_j^t}\right)\bar{q}_i \text{ and}
	\end{align*}
	\begin{align*}
	\widetilde{q}_i(t) & \leq \sum_{k=1}^K \frac{w_kq_{i,k}}{\sum_{j=1}^Nv_{\sigma(j)}d_j^t}\left(\sum_{j=1}^N  v_{\sigma(j)}(\max_{1\leq k\leq K}a_{i,k} +d_j^t)\right)\leq\left(1+\frac{\sum_{j=1}^N v_{\sigma(j)} \max_{1\leq k\leq K}a_{j,k}}{\sum_{j=1}^N v_{\sigma(j)}d_j^t}\right)\bar{q}_i.
	\end{align*}
	As a result, the bounds for $\widetilde{Q}_i(t)$ $(\forall i\in[1,N], t \in \mathbb{N})$ are given by
	\begin{align*}
	\left(1+\frac{\sum_{j=1}^N v_{\sigma(j)} \max_{1\leq k\leq K}a_{j,k}}{\sum_{j=1}^N v_{\sigma(j)}d_j^t}\right)\bar{Q}_i \geq \widetilde{Q}_i(t)  \geq \left(1-\frac{\sum_{j=1}^N v_{\sigma(j)} \max_{1\leq k\leq K}a_{j,k}+\max_{1\leq k \leq K}z_k}{\sum_{j=1}^N v_{\sigma(j)}d_j^t}\right)\bar{Q}_i.			
	\end{align*}
	Notice that these bounds converge to $\bar{Q}_i$, so they could be arbitrary close to $\bar{Q}_i$ by choosing a sufficiently large $d^t$ vector.
	
	\noindent To conclude the proof, we need to estimate the total number of
	purchases $\hat{d}_{tot}$ that guarantees that
	\[
	\forall t>t^*: \widetilde{Q}_{i^*}(t)>\widetilde{Q}_{i^{**}(t)}
	\]
	where $t^*$ is the time period in which the total number of purchases
	becomes $\hat{d}_{tot}=\sum_{i\in[N]}d_i^{t^*}$. The value
	$\hat{d}_{tot}$ and its associated vector of purchases $d^{t^*}$ must
	satisfy the following condition
	\begin{equation}
	\label{dtot_condition}
	\Delta \bar{Q} >\frac{\sum_{j=1}^N v_{\sigma(j)} \max_{1\leq k\leq K}a_{j,k}(\bar{Q}_{i^*}+\bar{Q}_{i^{**}})+\bar{Q}_{i^*}\max_{1\leq k \leq K}z_k}{\sum_{j=1}^N v_{\sigma(j)}d_j^{t^*}}.
	\end{equation}
	To verify inequality \eqref{dtot_condition}, it suffices to choose $\hat{d}_{tot}$ to satisfy
	
	\begin{equation}
	\hat{d}_{tot}>\frac{\sum_{j=1}^N v_{\sigma(j)} \max_{1\leq k\leq K}a_{j,k}(\bar{Q}_{i^*}+\bar{Q}_{i^{**}})+\bar{Q}_{i^*}\max_{1\leq k \leq K}z_k}{\max_j v_{\sigma(j)}\Delta\bar{Q}}.
	\end{equation}
	Since $\hat{d}_{tot}$ can be as large as desired, with the previous condition we guarantee the validity of Eq. \eqref{dtot_condition}.
	
	We have just shown that the conditions of Lemma \ref{qsuc} are satisfied, we can now apply it using ranking policy $\sigma$ to
	prove that the model goes to a monopoly for item $i^*$, which
	maximizes the product of its visibility and its weighted average
	quality, i.e., $v_{\sigma(i)} \bar{q}_i$.
\end{proof}	
\begin{proof}[Proof of Corollary \ref{cor1}]
	From Theorem \ref{monopoly_alltogether}, a MMNL model goes to a
	monopoly for the item $i$ that maximizes $v_{\sigma(i)}
	\bar{q}_i$. When the quality ranking is used, the product
	$i^*$ that goes to a monopoly is \[i^*=\argmax_i v_{\sigma(i)} \bar{q}_i = \argmax_i(\bar{q}_i).\]		
\end{proof}		
\begin{proof}[Proof of Theorem \ref{bound-nsi}]
	
	To prove this result, we need the
	following property. Let $A \in \mathbb{R}^{N\times K}_{\geq 0}$, then
	
	\begin{equation}
	\label{simple_matrix_inequality}
	\sum_{k=1}^K \max_{1 \leq i \leq N} a_{i,k} \leq K \max_{1 \leq i \leq N} \sum_{k=1}^K a_{i,k}
	\end{equation}
	where $a_{i,k}\in A$. Its proof follows from the following argument; $K\max_{1\leq i \leq N}\sum_{k=1}^K a_{ik}=K\max_{1\leq i \leq N}\sum_{k=1}^K |a_{ik}|$ $=K \left\Vert A \right\Vert_\infty=K \sup_{v\in \mathbb{R}^{K\times 1}: \Vert v \Vert_\infty = 1} \Vert Av\Vert_\infty \geq  K\left\Vert A \mathbb{I}^{K\times 1}\frac{1}{K} \right\Vert_\infty = \left\Vert A \mathbb{I}^{K\times 1} \right\Vert_\infty=\sum_{i=1}^N\sum_{k=1}^K |a_{ik}|\geq \sum_{k=1}^K \max_{1 \leq i \leq N} |a_{ik}|=\sum_{k=1}^K \max_{1 \leq i \leq N} a_{ik}$, where $\mathbb{I}^{K\times 1}$ is an all-one vector of dimension $K\times 1$.
	
	At the limit, the probability that an item is purchased under the
	average quality ranking with the popularity signal is given by
	\begin{equation}
	P_{AQGSI}=\max_{1\leq i \leq N} \bar{q}_i.
	\end{equation}
	When no popularity signal is shown, this probability becomes
	\begin{equation}
	\label{global-NSI}
	P_{AQNSI}=\sum_{k=1}^K w_k\sum_{i=1}^N q_{i,k}\frac{v_{\sigma(i)}a_{i,k}}{\sum_{j=1}^Nv_{\sigma(j)}a_{j,k}+z_k}.
	\end{equation}
	We can easily bound $P_{AQNSI}$ as follows:
	\begin{equation}
	0\leq\sum_{k=1}^K \min_{1\leq i \leq N}(w_k q_{i,k})\frac{\sum_{i=1}^Nv_{\sigma(i)}a_{i,k}}{\sum_{j=1}^Nv_{\sigma(j)}a_{j,k}+z_k}\leq P_{AQNSI}\leq\sum_{k=1}^K \max_{1\leq i \leq N}(w_k q_{i,k})
	\end{equation}
	and hence, by Inequality \eqref{simple_matrix_inequality},
	\begin{equation}
	0\leq\frac{P_{AQNSI}}{P_{AQGSI}}\leq K.
	\end{equation}
\end{proof}

\begin{proof}[Proof of Proposition \ref{bound-nsi-tightness}]
	Consider first the upper bound. Choose a MMNL model where $z=0$, $K=N$, the
	quality matrix is diagonal with a value of $1$ for the first element
	and $1-\epsilon$ for all others, the appeal matrix is the identity,
	and the classes have the same weights $w_i=\frac{1}{K}$. Then,
	\begin{align*}
	P_{AQNSI} =\sum_{1\leq i \leq N} \frac{1}{K}\left(1-\epsilon(1-\delta_{i1})\right) \text{ and } P_{AQGSI}=\frac{1}{K},
	\end{align*}
	where $\delta_{ij}$ is the Kronecker delta, thus
	\begin{equation}
	\lim_{\epsilon \to 0} \frac{P_{AQNSI}}{P_{AQGSI}}= \lim_{\epsilon \to 0} \sum_{1\leq i \leq N} (1-\epsilon\delta_{i1})= \lim_{\epsilon \to 0} (K-\epsilon (K-1)) = K.
	\end{equation}
	
	\noindent
	Consider now the lower bound. Choose a MMNL model where $z=0$, $K=N$ with the
	same quality matrix as before, the same weights, and an appeal matrix
	filled with ones except in its diagonal where each element has a value of
	$\epsilon_A$. Then,
	\begin{align*}
	P_{AQNSI} = \sum_{1\leq i \leq N} \frac{1}{K}\left(1-\epsilon(1-\delta_{i1})\right)\frac{v_{\sigma(i)}\epsilon_A}{\epsilon_Av_{\sigma(i)}+\sum_{j\neq i} v_{\sigma(j)}} \text{ and } P_{AQGSI} = \frac{1}{K},
	\end{align*}
	thus
	\begin{equation}
	\lim_{\epsilon_A \to 0} \frac{P_{AQNSI}}{P_{AQGSI}}= \lim_{\epsilon_A \to 0} \sum_{1\leq i \leq N}\left(1-\epsilon(1-\delta_{i1})\right)\frac{v_{\sigma(i)}\epsilon_A}{\epsilon_Av_{\sigma(i)}+\sum_{j\neq i} v_{\sigma(j)}} = 0.
	\end{equation}
	
\end{proof}

\begin{proof}[Proof of Theorem \ref{thm_bound}]
	By Theorem \ref{monopoly_alltogether}, we have
	\begin{equation*}
	P_{AQGSI} \doteq \lim_{t\rightarrow\infty} P^t_{AQGSI} = \max \{ \bar{q}_{1}, \bar{q}_{2}, \hdots, \bar{q}_{N}\}.
	\end{equation*}
	
	\noindent
	As mentioned earlier, for the segmented quality ranking, each segment is
	independent from each other and all of them will converge to a
	monopoly for the product with the highest quality in that class.
	We have that
	
	\begin{equation*}
	P_{SQSSI} \doteq \lim_{t\rightarrow\infty} P^t_{SQSSI} = \sum_{k=1}^K w_k \max_{i}q_{i,k}.
	\end{equation*}
	
	\noindent
	As a result,
	\begin{align*}
	\label{Psat}
	\frac{P_{SQSSI}}{P_{AQGSI}} & =\frac{\sum_{k=1}^K w_k \max_{1\leq i \leq N} q_{i,k}}{\max_{1\leq i \leq N}\sum_{k=1}^K w_k  q_{i,k}}
	=\frac{\sum_{k=1}^K  \max_{1\leq i \leq N} w_k q_{i,k}}{\max_{1\leq i \leq N}\sum_{k=1}^K w_k  q_{i,k}}.
	\end{align*}
	
	\noindent
	The lower bound is obviously valid and the upper bound follows from
	inequality \eqref{simple_matrix_inequality}.
\end{proof}

\begin{proof}[Proof of Proposition \ref{tightness_segmentation}]
	Consider a model with $K$ items and $K$ consumer classes. Without
	loss of generality, let the segment 1 be the segment with the
	lowest weight, i.e., $w_1\leq w_k \forall k\in[K]$. Then, for any
	set of positive appeals in each class, define the elements
	$q_{i,k}$ as follows:
	\begin{equation*}
	q_{i,k}= \left\{
	\begin{array}{l l}
	\frac{\min_{j\in[K]}{w_j}}{w_1} & \quad \text{if $i=k=1$}\\
	\frac{\min_{j\in[K]}{w_j}}{w_k}-\epsilon & \quad \text{if $i=k\neq 1$} \\
	0 & \quad \text{otherwise}
	\end{array} \right.
	\end{equation*}
	\noindent
	where $\epsilon$ is a positive number ensuring that the model is
	tie-breaking for the quality rankings. Then
	\begin{equation*}
	\lim_{\epsilon \to 0} \frac{P_{SQSSI}}{P_{AQGSI}}= \lim_{\epsilon_A \to 0} \frac{K\min_{j\in[K]}{w_j}-\epsilon\sum_{k=2}^K w_k}{\min_{j\in[K]}w_j} = K.
	\end{equation*}
\end{proof}

\begin{proof}[Proof of Theorem \ref{AQNSIvsSQSSI_bound}]
	Using the result from \cite{van2016aligning}, each segment under the SQSSI ranking will go to a monopoly of the product that maximizes its quality. If this product is not unique, the market will be shared between such products. Thus, the ratio between the asymptotic purchase probability of the average quality ranking without social influence and its segmented version with segmented social influence is given by
	\begin{equation*}
	\frac{\lim_{t\to \infty}P^t_{AQNSI}}{\lim_{t\to \infty}P^t_{SQSSI}}=\frac{\sum_{k=1}^K w_k \sum_{i=1}^N q_{i,k}\frac{v_{\sigma(i)}a_{i,k}}{\sum_{j=1}^N v_{\sigma(j)}a_{j,k}+z_k}}{\sum_{k=1}^K w_k \max_{1\leq i \leq N}q_{i,k}}.
	\end{equation*}
	
	Since $\sum_{i=1}^N q_{i,k}\frac{v_{\sigma(i)}a_{i,k}}{\sum_{j=1}^N v_{\sigma(j)}a_{j,k}+z_k} \leq \max_{1\leq i \leq N}q_{i,k}$, we have that $\frac{\lim_{t\to \infty}P^t_{AQNSI}}{\lim_{t\to \infty}P^t_{SQSSI}}\leq 1$.
	
	A similar argument can be made for the second comparison:
	\begin{equation*}
	\frac{\lim_{t\to \infty}P^t_{SQNSI}}{\lim_{t\to \infty}P^t_{SQSSI}}=\frac{\sum_{k=1}^K w_k \sum_{i=1}^N q_{i,k}\frac{v_{\sigma_k(i)}a_{i,k}}{\sum_{j=1}^N v_{\sigma_k(j)}a_{j,k}+z_k}}{\sum_{k=1}^K w_k \max_{1\leq i \leq N}q_{i,k}}\leq 1.
	\end{equation*}
\end{proof}

\section*{Supplementary material for ``Market Segmentation in Online Platforms"}
\subsection*{Supplementary Appendix A: The impact of misclassification}
In this appendix we consider scenarios in which the firm is not always able to identify correctly the corresponding consumer segment. Misclassification may occur, for example, when there is no historical data about the incoming consumer and it has to be assigned to one of the classes solely based on basic information provided by some internet marketing companies (e.g. keywords searched, geographical region, sex, age, etc). We analyze what is the impact of having some classification errors under some mild model assumptions. First, we provide some theoretical results about the convergence under our misclassification model. Finally, using a numerical experiment, we illustrate what is the impact in performance as a function on the misclassification rate.

We begin describing the model extension. Suppose that every time the system has an arriving customer of segment $l\in [K]$ (this occurs with probability $w_l$), there exists a probability $\alpha_{lk}$ that the consumer is recognized as segment $k\in [K]$. If a consumer is recognized by the system (or classifier) as segment $k$ consumer (even if it is not really from segment $k$), it is said that the consumer is \emph{observed} in segment $k$. Similarly to the policy {\tt SQSSI} let {\tt SQSSIM} denote the ranking policy of classifying consumers (but now misclassification errors occur) and then providing to the consumers observed as segment $k$ the quality ranking of consumer segment $k$ and updating the popularity signal locally in each segment.

Since {\tt SQSSIM} has misclassification, the benefits of segmentations might be deteriorated depending on how often segments are mistakenly recognized. Under {\tt SQSSIM}, which product will become the most popular in each of the segments? To answer this question, we rely on the important result obtained in Theorem \ref{monopoly_alltogether}. Namely, it provides the long term convergence of using a static ranking with global social influence among all consumer segments. The following corollary identifies the product in each segment that will become the most popular in the long term.

\begin{corollary} \label{mistake_cor}
	Suppose the solution to $\argmax_{1\leq i\leq N}v_{\sigma_k(i)}\sum_{l=1}^K w_l\alpha_{lk}q_{i,l}$ is unique for segment $k$ where $\sigma_k(\cdot)$ denotes a static ranking associated to segment $k$. Then, under the {\tt SQSSIM} policy, the product $i^*_k$ that converges to a monopoly for the observed segment $k$ is:
	\begin{equation}\label{eq:mistakes_conv}
	i^*_k=\argmax_{1\leq i\leq N}v_{\sigma_k(i)}\sum_{l=1}^K w_l\alpha_{lk}q_{i,l}.
	\end{equation}
\end{corollary}
\proof Consumers observed in segment $k$ may belong to $K$ different segments due to classifications mistakes, thus, each observed segment $k$ can be seen as a MMNL. Given a consumer is observed in segment $k$, the probability that this customer belongs to $l\in [K]$ is given by $w_l\alpha_{lk}$. Then, the purchase probability of product $i$ for a consumer observed in segment $k$, with ranking $\sigma_k$ and purchase vector $d^{t,k}$ is given by
\begin{equation*}
P_i^k(\sigma_k,d^{t,k}) = \sum_{l=1}^K \Big( w_l\alpha_{lk} \cdot
\Big(v_{\sigma_k(i)}\frac{ (a_{i,l} + d_i^{t,k})}{\sum_{j=1}^N
	v_{\sigma_k(j)} (a_{j,l} + d_j^{t,k})+z_l} \cdot q_{i,l} \Big) \Big).
\end{equation*}
This probability resembles Eq. \eqref{eq:thm2_main} in the proof of Theorem \ref{monopoly_alltogether}. Thus, we know that each observed segment $k$ of consumers converges to a monopoly for product \[i^*_k=\argmax_{1\leq i\leq N}v_{\sigma_k(i)}\sum_{l=1}^K w_l\alpha_{lk}q_{i,l}.\]\endproof	

Now that we have found the long term behavior of the model with classifying errors, we analyze the impact of errors in market efficiency. From now on we assume that the probability of committing a mistake in classifying a segment is the same for every segment, and is equally likely to identify it as any other segment. Then we may define the mistake probabilities with two parameters, $\alpha_0$ and $\beta_0$:
\begin{equation}\label{eq_alpha}
\alpha_{lk}=\begin{cases} \alpha_0 \quad\text{ if } l=k \\
\frac{\beta_0}{K-1} \quad\text{ otherwise}
\end{cases}.
\end{equation}
Equation \eqref{eq_alpha} means that every customer has a probability $\alpha_0$ of being recognized as their correct segment, while a probability $\beta_0$ that  the consumer's segment is mistaken for another one. Naturally, $\alpha_0=1-\beta_0$.

Similarly to the definition of $P_{SQSSI}^t$, let $P_{SQSSIM}^t$ denote the probability of a purchase at time $t$ when the firm applies the segmented quality ranking with the local popularity signal $d_k^t$ under a classifier with errors. We are now ready to prove the following:

\begin{theorem}
	Assume that the average quality ranking and its segmented version are tie-breaking for a MMNL model with mistake probability matrix $\alpha$ given by \eqref{eq_alpha}, and that $\alpha_0>\frac{\beta_0}{K-1}$ (the probability than an observed segment is classified correctly is greater than the probability of misclassifying it with any other segment). Then,		
	\begin{equation}
	\frac{K}{K-1} \beta_0 \leq \lim_{t\rightarrow \infty}\frac{P^t_{SQSSIM}}{P^t_{AQGSI}}\leq K\alpha_0
	\end{equation}
	
\end{theorem}
\proof
As mentioned earlier, when we have segmentation, each segment is independent from each other and every recognized segment $k$ will converge to a monopoly for the product $i^*_k$ given by Eq. \eqref{eq:mistakes_conv}. If each observed segment $k$ is ranked according to the new qualities $\hat{q}_{i,k}\doteq\sum_{l=1}^K w_l\alpha_{lk}q_{i,l}$, then

\begin{flalign*}
\lim_{t\rightarrow \infty}P^t_{SQSSIM}&=\sum_{l=1}^K \sum_{k=1}^K w_l\alpha_{lk}q_{i^*_k,l}=\sum_{k=1}^K  \max_{1\leq i\leq N} \sum_{l=1}^K w_l\alpha_{lk}q_{i,l} \\
&\leq \sum_{k=1}^K [\max_{1\leq l \leq K}\alpha_{lk}] \max_{1\leq i\leq N} \sum_{l=1}^K w_lq_{i,l}= \sum_{k=1}^K \max\lbrace{\alpha_0,\frac{\beta_0}{K-1}\rbrace} \max_{1\leq i\leq N} \bar{q}_i\\
&=K \alpha_0 \max_{1\leq i\leq N} \bar{q}_i,
\end{flalign*}
and
\begin{flalign*}
\lim_{t\rightarrow \infty}P^t_{SQSSIM}&=\sum_{l=1}^K \sum_{k=1}^K w_l\alpha_{lk}q_{i^*_k,l}=\sum_{k=1}^K  \max_{1\leq i\leq N} \sum_{l=1}^K w_l\alpha_{lk}q_{i,l} \\
&\geq \sum_{k=1}^K [\min_{1\leq l \leq K}\alpha_{lk}] \max_{1\leq i\leq N} \sum_{l=1}^K w_lq_{i,l}= \sum_{k=1}^K \min\lbrace{\alpha_0,\frac{\beta_0}{K-1}\rbrace} \max_{1\leq i\leq N} \bar{q}_i\\
&=K \frac{\beta_0}{K-1} \max_{1\leq i\leq N} \bar{q}_i.
\end{flalign*}

Since $\lim_{t\rightarrow \infty}P^t_{AQGSI}=\max_{1\leq i \leq N} \bar{q}_i$, we finally have that
\begin{equation*}
\frac{K}{K-1} \beta_0 \leq	\lim_{t\rightarrow \infty}\frac{P^t_{SQSSIM}}{P^t_{AQGSI}}\leq K\alpha_0.
\end{equation*}	
\endproof

To illustrate the effects of classifications mistakes in the market, we perform a simulation for Scheme 2 under the ranking policy SQSSIM with different values of $\alpha_0$. We also plot the ranking policy AQGSI in the same graph, the results are shown in Figure \ref{fig:mistakes_scheme1}. As expected, the performance of the ranking policy SQSSIM decreases as the percentage of correct consumer segment classification ($\alpha_0$) decreases. Furthermore, it is interesting to see that the AQGSI ranking policy outperforms SQSSIM with an $\alpha=0.8$ or less. The managerial insight is that segmentation in this setting is better than showing a single ranking, but only as long as the misclassification errors are relatively small. In cases where $\alpha=0.8$ or lower, market segmentation is harmful.	
\begin{figure}[t]
	\centering
	\includegraphics[trim = 1.3in 3in 1.5in 3.1in , clip, width=8cm]{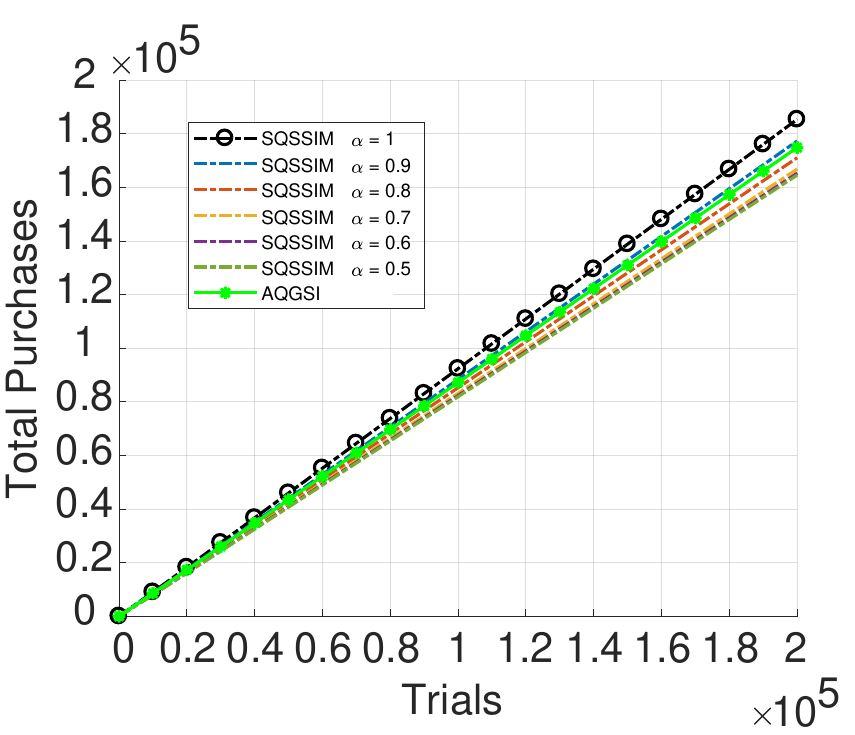}
	\hspace{0.5cm}
	\caption{The Number of Purchases over time for various SQSSIM rankings with different $\alpha_0$ values, and the AQGSI ranking for Scheme 2. The x-axis represents the number of items tried and the y-axis represents
		the average number of purchases over all experiments. We can see that only when the classifications errors are small ($\alpha \geq 0.8$), the segmentation policy outperforms the global ranking policy (average quality ranking AQGSI).}
	\label{fig:mistakes_scheme1}
\end{figure}

\subsection*{Supplementary Appendix B: 2-Swap, a Performance Ranking Heuristic}

In order to assess the average quality ranking policy (AQGSI), we performed computational experiments and compare it to a substantially more computationally expensive heuristic: the 2-swap heuristic. In all our experiments using the different schemes, the results obtained with the 2-swap heuristic were not significantly better than those obtained with the average quality ranking (AQGSI). Figure \ref{fig:2swap-sch1} shows the average total purchases over $400$ simulations as a function of the number of trials using the both methods. For completeness we also included the segmented ranking policy SQSSI.

\begin{figure}[h]\centering
	\includegraphics[trim = 1.1in 3.3in 1.7in 3.3in, clip, width=8.9cm]{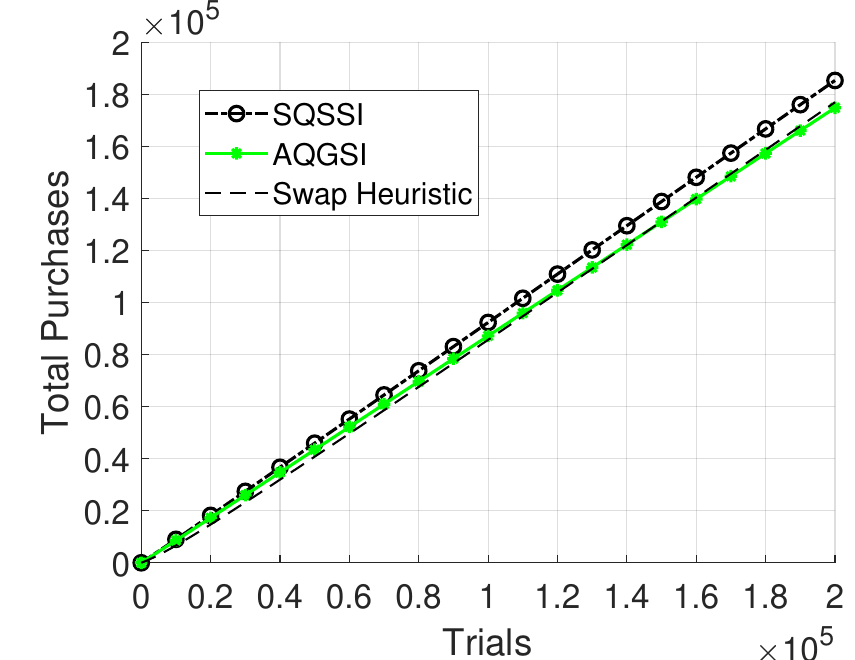}
	\caption{\label{fig:2swap-sch1} Average total purchases vs trials of $400$ simulations under SQSSI, SQSSI and the 2-swap heuristic with global social influence using Scheme 2.}
\end{figure}
\pagebreak

\subsection*{Supplementary Appendix C: Revenue Maximization}

In this appendix we analyze an extension of the model in which each product $i\in [N]$ generates a profit $r_i$ to the platform owner. Consumers are aware of the product prices and those are taken into account in the consumer choice probabilities of our model (Equation \eqref{MMNL}).

Since now the firm is interested in maximizing profit, we need to modify the previous objective (Equation \eqref{performance_ranking-MMNL}) to the following one:
\begin{align}
\label{sup_performance_ranking}
\Pi^{PR} & = \max_{\sigma \in S_N} \bigg\{ \sum_{k=1}^K \Big( w_k  \cdot \sum_{i=1}^N \Big(\frac{v_{\sigma(i)} (a_{i,k} + d_i^t)}{\sum_{j=1}^N v_{\sigma(j)} (a_{j,k} + d_j^t)+z_k}  \cdot r_i q_{i,k} \Big) \Big) \bigg\}.
\end{align}

Observe that when $r_i=1$ for all $i=1,\hdots,N$ the above problem reduces to the original setting in which the firm tries to maximize the expected number of purchases. This means that our NP-hardness result (Theorem \ref{theorem_hardness}) for the original model still holds for this new setting. Moreover, Theorem \ref{monopoly_alltogether} (the monopoly convergence) \ref{cor1} still hold, since consumer preferences under the Mixed MNL already account for the product's prices/revenues and the market dynamics is only dependent on those mixed MNL preferences (as previously). What is interesting to note is that in this model extension, the average quality ranking may lead to a monopoly for a product that doesn't yield the greatest revenue rate. This is because the product with the largest average quality, may not be the one that generates the largest expected revenue conditional on a trial. This lead us to analyze a new ranking heuristic, the ``Average Expected Revenue Ranking", which ranks products based on their expected revenue conditional on a trial, namely $\sum_k w_k q_{i,k}\times r_i$. To obtain analytical results, we analyze this policy under an important assumption: the average quality ordering is the same as the revenue ordering. This assumption is reasonable in multiple settings: it says that higher quality products have a higher price. With this assumption, the average expected revenue ranking is the same as the average quality ranking, and thus, from Theorem \ref{monopoly_alltogether}, the market converges to monopoly for the product with the largest expected revenue conditional on a trial (asymptotic optimality). In addition, for the Average Expected Revenue Ranking, Theorem \ref{bound-nsi} still holds. The Average Expected Revenue Ranking with the same ordering between average product qualities and product revenues, can perform up to $K$ (the number of consumer segments) times better, or arbitrarily worst, by not showing the popularity signal. The proof is straightforward from Theorem \ref{bound-nsi}'s proof, in which $q_{i,k}$ is exchanged by $q_{i,k}r_i$ each time.

For Theorem \ref{thm_bound} to hold with the Average Expected Revenue Ranking, we need a stronger assumption. If the average quality ordering is just the same as the revenue ordering, it could happen that for some consumer segments product qualities and revenues are not ordered the same way. Thus, ranking by expected revenue conditional on a trial for each segment might not achieve asymptotic optimality in each consumer segment. In order to guarantee asymptotic optimality of each consumer segment we need that the quality ordering for each them is the same as the revenue ordering. Thus, the results from Theorem \ref{thm_bound} hold only when all consumer segment product qualities follow the same ordering as the product revenue ordering.

\subsection*{Supplementary Appendix D: Assortment Optimization}
In this section we analyze the extension in which in each time period the platform owner chooses a subset $S \subseteq [N]$ of products to show to consumers, as well as a ranking $\sigma_S$ among them. We particularly focus in the case where the platform owner has no information on consumers segments, and thus, needs to display the same assortment of products to all consumers. If there is only one consumer segment, this problem can be solved efficiently (see \citet{abeliuk2016assortment} and \citet{sumida2019revenue}). However, when there are at least two segments, the problem becomes a generalization of the classical assortment optimization under the latent class MNL model which is already NP-hard (see \citet{bront2009column} and \citet{rusmevichientong2014assortment}).

Given the impossibility of finding efficiently an optimal assortment, we study a heuristic which we called the {\it Average Quality Threshold Heuristic}. This heuristic first ranks products by their average quality, so let product $i$ denote the product with the $i^{th}$ highest average quality. Then, at each time period, it chooses how many products to show with the following condition: if it shows $k$ products, those must be $\{1,2\hdots,k\}$ and they should be ranked in the same order: higher quality products are placed in positions with higher visibility. This heuristic is computationally more intensive than the standard average quality, since for each time period, we need to evaluate $N$ different scenarios and choose the best one. Nevertheless, it is still computational practical since it takes at most $O(N^3)$ time.


Observe that the number of products shown by the Average Quality Threshold Heuristic is sensitive to the values associated to the outside option $z_k$'s. In one extreme, very large values of the outside option (in comparison to the value associated to the products) means that the products offered by the platform owner face a rather weak cannibalization. In those scenarios, adding an extra product to an assortment is likely to be beneficial since it will probably increase overall sales. On the other extreme, when the outside option values are very small, products offered by the firm face a strong cannibalization: in these scenarios it is likely that offering only a few products is the optimal strategy.

We performed the same computational experiments as those in Section \ref{section:simulations} on Scheme $1$ and Scheme $2$, but using the Average Quality Threshold Heuristic and varying the outside option value ($z$). A large outside option value means that at the early stages a higher number of consumers will decide not to try a product at all. In the long run, the effect of the outside option diminishes as some products become popular and their overall appeal increases. Figures \ref{fig:optimalthreshold_vs_t_1} and \ref{fig:optimalthreshold_vs_t_2} show the optimal number of products that are shown at each period (on average) for Schemes 1 and 2 respectively. As expected, observe that the optimal number of products shown increases when the outside option increases. However, as the number of purchases increase (and the social influence signal becomes stronger), the offer sets shown tend to reduce in size to exhibit only the products with the highest average quality.

In all our computational experiments, the offer sets offered by the average quality threshold heuristic are always reduced in size or they stay the same as times goes by. This suggests to us that even under the average quality threshold heuristic, which is a dynamic ranking policy, the market will also converge to a monopoly for the product with the highest weighted average quality. We leave this conjecture below.

\begin{conjecture}\label{conjecture}{(Asymptotic optimality of the Average Quality Threshold Heuristic)}
	Whenever the average quality threshold heuristic is used, the market goes to a monopoly for the product with the highest weighted average quality.
\end{conjecture}

\begin{figure}[t]
	\subfigure{\label{fig:optimalthreshold_vs_t_1}
		\includegraphics[clip, width=8cm]{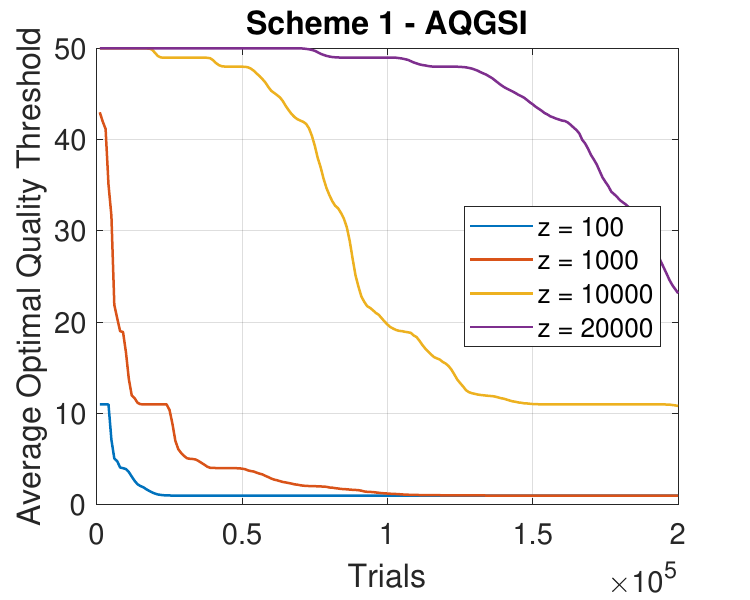}}
	\subfigure{\label{fig:optimalthreshold_vs_t_2}
		\includegraphics[clip, width=8cm]{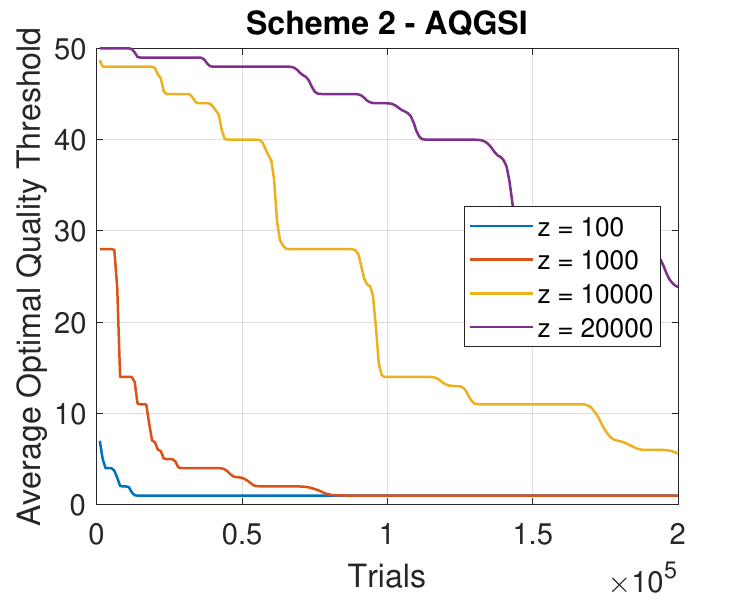}}
	\caption{Average threshold for the AQTGSI policy versus trials. On the left (a) we have it for Scheme 1, while on the right (b) we have it for Scheme 2.}
\end{figure}

We conclude this section by assessing the performance of the Average Quality Threshold Heuristic with Global Social Influence (AQTGSI) and we compare against the performance of the Average Quality with Global Social Influence (AQGSI). Figures \ref{fig:dl_vs_t1} and \ref{fig:dl_vs_t2} display the average number of purchases under AQTGSI for Schemes $1$ and $2$ respectively. We can observe that more sales occur in settings where the outside option value is small, but this difference is reduced as time goes by. Figures \ref{fig:percentage1} and \ref{fig:percentage2} show the percentage improvement (or deterioration) that AQTGSI has over AQGSI for Schemes $1$ and $2$ respectively. Our experiments show that these two policies have approximately the same performance for all the outside option values we tried{\color{black}, with less than $0.75\%$ improvement in 200,000 trials.}

\begin{figure}[t]
	\subfigure{
				\includegraphics[clip, width=8cm]{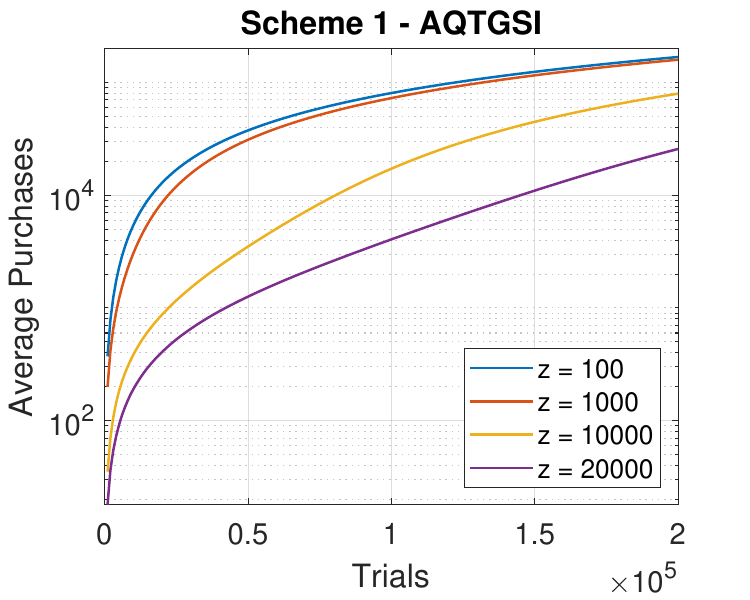}\label{fig:dl_vs_t1}}
		\subfigure{
			\includegraphics[clip, width=8cm]{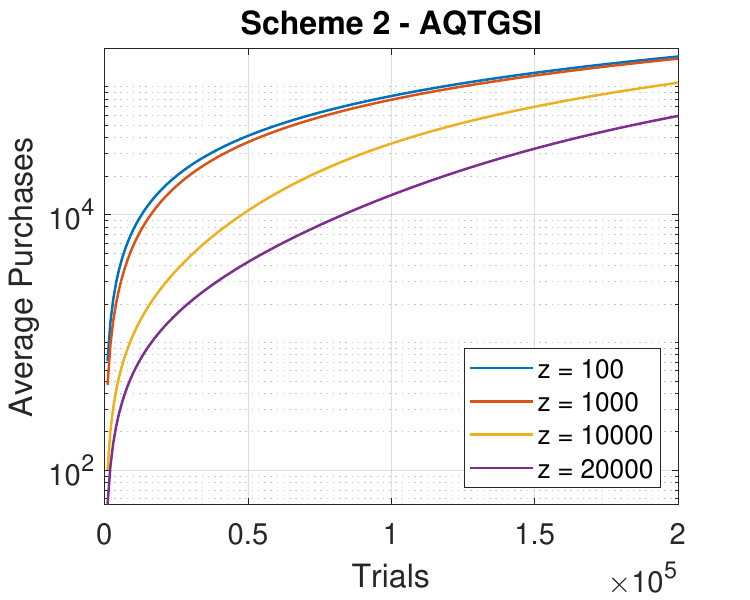}\label{fig:dl_vs_t2}}
	\caption{Number of purchases versus trials using the AQTGSI policy for scheme $1$ on the left (a) and Scheme $2$ on the right (b).}
\end{figure}

\begin{figure}[t]
	\subfigure{
		\includegraphics[clip, width=8cm]{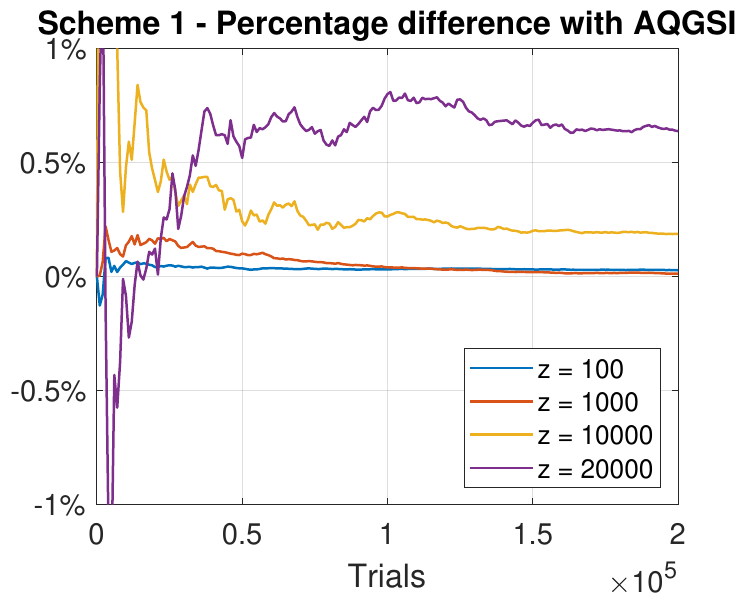}		\label{fig:percentage1}}
		\hspace{0.5cm}
		\subfigure{
			\includegraphics[clip, width=8cm]{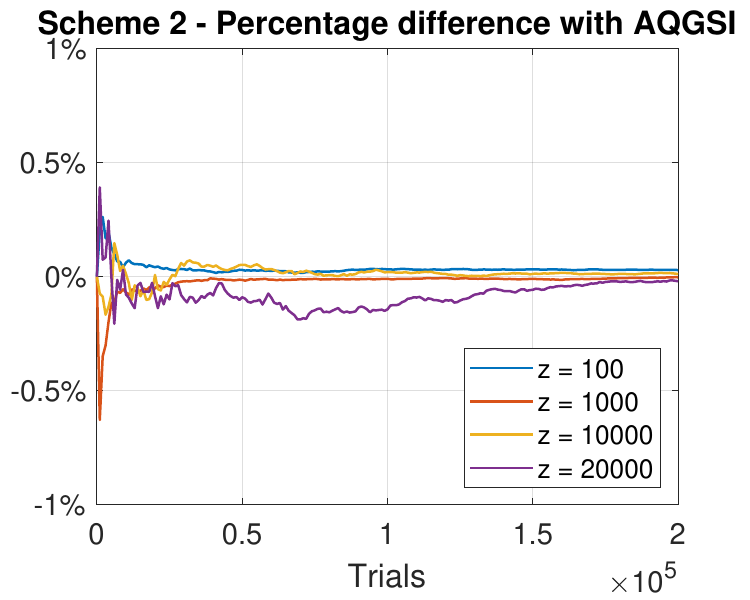}	\label{fig:percentage2}}
	\caption{Percentage difference in the number of purchases between AQTGSI and AQGSI for Scheme $1$ on the left (a) and Scheme $2$ on the right (b).}
\end{figure}

\end{document}